\documentclass[reqno]{amsart}
\usepackage{amssymb}
\usepackage[dvips]{epsfig}
\usepackage{graphicx}
\usepackage{color}
\usepackage{amsmath}
\usepackage{amssymb}
\usepackage{amsfonts}
\usepackage{amsthm}

\newcommand{\R}{\mathbb R}

\newcommand{\Z}{\mathbb Z}

\newcommand{\C}{\mathbb C}

\renewcommand{\v}{\mathfrak{v}}

\newtheorem{thm}{Theorem}[section]
\newtheorem{lem}[thm]{Lemma}
\newtheorem{prop}[thm]{Proposition}
\newtheorem{cor}[thm]{Corollary}
\theoremstyle{remark}
\newtheorem{rem}{\bf Remark}[section]
\theoremstyle{definition}
\newtheorem{defn}[thm]{Definition}

\newtheorem{ex}[thm]{Example}

\numberwithin{equation}{section}
\usepackage[colorlinks=true,pdfstartview=FitV,linkcolor=magenta,citecolor=cyan]{hyperref}

\begin{document}

\title[Nash-Moser Localization]{Localization for Almost-Periodic Operators with Power-law Long-range Hopping: A Nash-Moser Iteration Type Reducibility Approach}
\author[Y. Shi]{Yunfeng Shi}
\address[Y. Shi] {School of Mathematics,
Sichuan University,
Chengdu 610064,
China}
\email{yunfengshi@scu.edu.cn}

\date{\today}


\begin{abstract}

In this paper we develop a Nash-Moser   iteration type reducibility  approach to prove the (inverse)  localization for some $d$-dimensional discrete almost-periodic operators with power-law  long-range hopping. We also provide a quantitative lower bound on the regularity  of the hopping. As an application,   some results of \cite{Sar82, Pos83, Cra83, BLS83} are generalized  to the power-law hopping case.
\end{abstract}

\maketitle

\maketitle
\section{Introduction}
 The study of the localization for noninteracting quantum particles in (pseudo) random media with \textit{finite-range} hopping has attracted great attention over the years since the seminal  work of Anderson (\cite{And58}). While models with \textit{finite-range} hopping work nicely in describing variety of materials, the {power-law} long-range hopping is often found in different physical systems, e.g., dipolar Frenkel exciton (\cite{Nab95}), nuclear spins in solid-state systems (\cite{ASK15})  and the quantum Kepler model (\cite{AL97}). Actually,   some  physical works have presented  unusual localization properties of both random (\cite{Rod03}) and quasi-periodic (\cite{Den19}) operators with power-law hopping.  Apart from the importance of the localization theory (of  power-law hopping models) itself, another key motivation comes from the  perspective of \textit{quantum suppression of chaos} related to the localization  in quantum chaos (\cite{FGP82,Izr90,SPK22})  which will be  elaborated below.  One of the main models  in quantum chaos is the so-called  $1$-dimensional quantum kicked rotor  which is given by
\begin{align}\label{qkrm}
\sqrt{-1}\frac{\partial \Psi(x, t)}{\partial t}= \left(\mathcal{H}_0+\check{\phi}(x)\sum_{n\in\Z}\delta(t-nT)\right)\Psi(x, t),
\end{align}
where
\begin{align*}
\ (x, t)&\in\R\times \R, \ T>0, \\
\mathcal{H}_0&=\frac{\partial^2 }{\partial x^2}\ {\rm or} \ \sqrt{-1}\frac{\partial}{\partial x},
\end{align*}
and $\check{\phi}$ is a potential of period $1$. The model \eqref{qkrm} with $\mathcal{H}_0=\frac{\partial^2 }{\partial x^2}$ and $\check{\phi}(x)=\cos(2\pi x)$ was first introduced by Casati \textit{et al} (\cite{Cas79}) as a quantum analogue of the usual Chirikov standard map. The  motion in  the quantum kicked rotor is generally  almost-periodic  even though  it is  typically chaotic in the classical one (\cite{FGP82,Bou02}).   This \textit{quantum suppression of chaos} feature was well understood  after the remarkable work of Fishman, Grempel and Prange (\cite{FGP82,GFP82}). By introducing an elegant transformation,  they reduced the quantum kicked rotor  to a  discrete operator $H$  of the form
\begin{align*}
(Hu)_n=\sum_{m\in\Z}\phi_{n-m}u_m+d_n(y)u_n,\ n\in\Z,\  y\in\R,
\end{align*}
where
\begin{align*}
\phi_n&=-\int_{0}^1\tan\pi\left(\frac{\check{\phi}(x)}{2}\right)e^{-2\pi\sqrt{-1} nx}\mathrm{d}x,\\
d_n(y)&=\tan\pi\left(y-n^2\frac{T}{2}\right)\ {\rm or}\ \tan\pi\left(y-n\frac{T}{2}\right)\ \mbox{depending on} \ \mathcal{H}_0.
\end{align*}
Moreover,  under the above reduction, the localization for $H$ could imply the almost-periodicity of solutions for the quantum kicked rotor \eqref{qkrm}, and then the absence of chaos (cf. \cite{FGP82,SPK22}).  For the analytic potential $\check{\phi}$, the operator $H$ admits an exponential  hopping, and \textit{all coupling} localization for $H$ has been obtained both in physics (\cite{GFP82}) and mathematics (\cite{FP84}) under certain non-resonant condition of $T/2$. However, the non-analytic (even singular) potential $\check{\phi}$ which yields a slowly decaying  hopping also appears naturally in some important physical models, e.g., the quantum Fermi accelerator (\cite{JC86}). In fact, the  work \cite{GW05} has provided numerical evidence for  power-law localization in quantum chaos with  the  singular potential  $\check{\phi}(x)=|x|^\alpha\mod 1$ ($\alpha> -1$), which induces a power-law hopping $|\phi_n|\sim |n|^{-1-\alpha}$. Note however that the analytic derivation of the localization  in \cite{GW05} relies on physical perspective of localization  for random operators with power-law hopping (cf. \cite{Rod03}).  In addition, the  on-site energy sequence $\{d_n(y)\}_{n\in\Z}$ is only pseudo-random and the localization properties of $H$ should depend  on arithmetic properties of $T/2$. Thus, a mathematically rigorous  treatment of the localization for quasi-periodic operators with (realistic) power-law  hopping becomes significantly relevant. We would also like to mention that certain quantum kicked rotors in higher dimensions (\cite{DF88}) or with a quasi-periodic potential (\cite{CGS89})  will give rise to quasi-periodic operators on the higher dimensional lattice. The present paper aims to develop a Nash-Moser iteration type reducibility approach to prove the (inverse)  localization for some $d$-dimensional almost-periodic operators with power-law  hopping. In particular, our result applies to some physical realistic hopping  $|\phi_n|\sim |n|^{-1-\alpha}$ for $\alpha>1$ (e.g., the case $\alpha=2$  corresponds to the $1$-dimensional Frenkel exciton hopping \cite{RMD00}).



More precisely, let us start with the  discrete operator
\begin{align}\label{lrop}
H=\varepsilon T_\phi+d_\mathbf{i}\delta_{\mathbf{ii}'},\ \mathbf{i} ,\mathbf{i}'\in\Z^d,\ \varepsilon\geq0,
\end{align}
where $(d_\mathbf{i})_{\mathbf{i}\in\Z^d}$ is a $d$-dimensional sequence and $T_\phi$ is the long-range hopping operator defined by
\begin{align*}
(T_\phi u)_\mathbf{i}&=\sum_{\mathbf{j}\in\Z^d}{\phi_{\mathbf{i}-\mathbf{j}}u_\mathbf{j}},\ u=(u_\mathbf{i})_{\mathbf{i}\in\Z^d}\in \ell^2({\Z^d})
\end{align*}
  with a symbol $\phi=(\phi_\mathbf{i})_{\mathbf{i}\in\Z^d}\in \R^{\Z^d}$. If we let $\phi_\mathbf{i}=\delta_{\mathbf{ie}}$ with $|\mathbf{e}|_1=\sum\limits_{i=1}^d|e_i|=1$, then \eqref{lrop} becomes the standard discrete Schr\"odinger operator and we write $T_\phi=\Delta$. In the present we focus on the almost-periodic sequences  $(d_\mathbf{i})_{\mathbf{i}\in\Z^d}$, and some special cases include $d_\mathbf{i}=\tan\pi(\mathbf{i}\cdot\vec\omega),\ d_\mathbf{i}=\exp{(2\pi\sqrt{-1}\mathbf{i}\cdot\vec\omega)},\  d_\mathbf{i}=(\mathbf{i}\cdot\vec\omega\mod 1)$ with $\mathbf{i}\cdot\vec\omega=\sum\limits_{v=1}^di_v\omega_v$, $\vec\omega\in\R^d.$ If $\varepsilon=0$, we know that the spectrum of $H$ is pure point and $\{\delta_\mathbf{i}\}_{\mathbf{i}\in \Z^d}$ forms a complete set of eigenfunctions. A natural question is  whether such localization  preserves or not for $\varepsilon\neq0$. It turns out that  this is quite a delicate problem and the localization depends  sensitively  on certain parameters associated to the operator (\cite{Sim82, Bou05, MJ17, Dam17}). However, it is an intuition  that $H$ could be ``diagonalizable'' and thus shows localization if the coupling $\varepsilon$ is small enough and $(d_\mathbf{i})_{\mathbf{i}\in\Z^d}$ is reasonably ``separated''. In this context the celebrated  Kolmogorov-Arnold-Moser (KAM) (\cite{Kol54,Arn63, Mos62a}) method becomes a good candidate to handle such issues.

Indeed, ever since Dinaburg-Sinai \cite{DS75} first introduced the KAM method  in the field of almost-periodic operators, it has become a very powerful tool to achieve both delocalization  and  localization (see e.g., \cite{Cra83, BLS83, Pos83, MP84, Sin86, Eli92, CD93, Eli97, AFK11} just for a few). Particularly, in the $1$-dimensional Laplacian hopping case, the KAM method can be modified to give reducibility of corresponding transfer matrix, and thus to obtain some delocalization results in the ``small'' potentials case. By the remarkable Aubry-Audr\'e duality \cite{AA80}, the delocalization at the small potential may imply localization of its duality at the large potential \cite{BLT83,JK16, AYZ17}.

An alternative method  is to ``diagonalize'' (or reduce) the infinite matrix $H$ directly  \cite{Cra83, BLS83, Pos83, CD93, Eli97}. This idea was first introduced by Craig \cite{Cra83}.  In  \cite{Cra83} Craig  performed an inverse spectral procedure relied on a modified KAM method, and obtained the existence of almost-periodic Schr\"odinger operators satisfying the Anderson localization. Let $D={\rm diag}_{\mathbf{i}\in\Z^d}(d_\mathbf{i})$. The core of the proof in \cite{Cra83} is to  find a unitary transformation and a diagonal operator $D'$ so that
\begin{align*}
Q^{-1}(\varepsilon\Delta+D')Q=D,
\end{align*}
where $Q$ is derived from  the limit of a sequence of invertible operators  in the KAM iteration steps. To get such transformations, Craig imposed some non-resonant condition on $(d_\mathbf{i})_{\mathbf{i}\in\Z^d}$, i.e.,
\begin{align}\label{dccond}
|d_\mathbf{i}-d_\mathbf{j}|\geq \Omega(|\mathbf{i}-\mathbf{j}|)\ {\rm for}\ \forall \  \mathbf{i}\neq \mathbf{j},
\end{align}
where $|\mathbf{i}|=\sup_{1\leq v\leq d}|i_v|$ and $\Omega:\ \R_{+}\to \R_{+}$ is some weight function which decays slower than the exponential one (\cite{Rus80}). The special case that $\Omega(t)=\gamma t^{-\tau}$ with some  $\tau>d, \gamma>0$  corresponds to the standard Diophantine condition. At every iteration step, some new diagonal operator will emerge, which would not satisfy the condition  \eqref{dccond}  in general. Then Craig  placed those new diagonal terms in $D'$, and as a result, only inverse spectral type results were obtained. Later, Bellissard-Lima-Scoppola (\cite{BLS83}) dealt with the direct problem  with
\begin{align}\label{lrcond}
|\phi_\mathbf{i}|\leq Ce^{-\rho|\mathbf{i}|}\ {\rm for\ some}\  C>0,\ \rho>0\ {\rm and}\ \forall\ \mathbf{i}\in\Z^d.
\end{align}
They observed that for some special almost-periodic  potentials, the condition \eqref{dccond} is stable under small perturbations. Then by using again the KAM method, they showed  for such potentials $(d_\mathbf{i})_{\mathbf{i}\in\Z^d}$, there exists a unitary operator $Q$ and some diagonal operator $D'$ so that for $|\varepsilon|\ll1$,
\begin{align}\label{blsmd}
Q^{-1}(\varepsilon T_\phi+D)Q=D'.
\end{align}
This implies in particular that $H$ has pure point spectrum. Since in this KAM procedure,  $Q\approx I$ (with $I$ being the identity operator) in the analytic norm, the eigenfunctions of $D'$ under the transformation $Q$ form a complete set of exponentially localized eigenfunctions of $H$. While the potentials of \cite{BLS83} seem restrictive, they contain in fact the Maryland potential (cf. \cite{GFP82}) and Sarnak's potential (cf. \cite{Sar82})  as special cases. Subsequently, P\"oschel \cite{Pos83} presented a general KAM approach (in the setting of translation invariant Banach algebras) and provided new examples of limit-periodic Schr\"odinger potentials having Anderson localization. P\"oschel's proof also requires the operator to  satisfy both \eqref{dccond} and \eqref{lrcond}.

All mathematical results as mentioned above concern operators with exponential long-range hopping. In this paper we try to generalize some results of \cite{Cra83,BLS83, Pos83} to the polynomial long-range hopping case, i.e.,
\begin{align*}
|\phi_\mathbf{i}|\leq |\mathbf{i}|^{-s}\ {\rm for\ some}\ s>0\ {\rm and}\ \forall\ \mathbf{i}\in\Z^d\setminus\{\mathbf{0}\}.
\end{align*}
We assume additionally $(d_\mathbf{i})_{\mathbf{i}\in\Z^d}$ satisfies the Diophantine condition (this is reasonable since we have a slower decay of the hopping). We again want to find invertible transformation $Q$  and some diagonal operator $D'$ so that \eqref{blsmd} holds true. The transformation $Q$ should  be the limit of some sequence of invertible  operators $Q_k$ along the iteration  steps  in \textit{some operator norm}.
In  \cite{BLS83, Pos83},  they introduced an exponential norm for the Toeplitz type operator $A$:
\begin{align*}
\|A\|_a=\sum_{\mathbf{k}\in\Z^d}|A_\mathbf{k}|e^{a|\mathbf{k}|},\  a>0,
\end{align*}
where $A_\mathbf{k}$ denotes the $\mathbf{k}$-diagonal of $A$ (see \eqref{sob} in the following for details). Thus at the $k$-th KAM  step, it needs to find  $Q_k$ and $D_{k-1}$ so that
\begin{align}\label{kamqk}
Q_k^{-1}\left(\varepsilon T_\phi+D+\sum_{i=1}^{k}D_{i-1}\right)Q_k=D+R_k,\ \|R_k\|_{\rho_k}=O(\varepsilon^{(\frac32)^k}),
\end{align}
where  $\inf\limits_{k\geq 0}\rho_k\geq \rho/2.$ The key point is to  determine $Q_k$ and $D_{k-1}$ which turn out to be the solutions of so-called homological equations. Since the small divisors difficulty, one  may lose  some  regularity at each iteration step. Fortunately, the condition \eqref{dccond} and the analytic norm permit a loss of order $\delta$ for \textit{arbitrary} $\delta>0$. This combined with the sup-exponential smallness of the Newton error will lead to the convergence of $Q_k$ and $\sum_{i=1}^{k}D_{i-1}$ in the $\|\cdot\|_{\rho/2}$ norm eventually, and then the exponential decay preserves! However, when we deal with the polynomial one, we should use the Sobolev type norm
\begin{align*}
\|A\|_s^2=\sum_{\mathbf{k}\in\Z^d}|A_\mathbf{\mathbf{k}}|^2\langle \mathbf{k}\rangle^{2s}, \ \langle \mathbf{k}\rangle=\max\{1, |\mathbf{k}|\},\ s>0.
\end{align*}
If we want to perform a similar  scheme as in \cite{BLS83, Pos83}, there comes a serious difficulty: Since we are in the polynomial case, the loss of regularity when solving the homological equations is  of order $\tau>d$ which is fixed at each step! We can imagine such iterations must fail in some finite steps (i.e., a loss of all regularities). This motivates us to employ instead a Nash-Moser iteration  type scheme.  Hence at the $k$-th step we would like to get
\begin{align*}
Q_k^{-1}\left(\varepsilon \sum_{i=1}^kT_{i-1}+D+\sum_{i=1}^{k}D_{i-1}\right)Q_k=D+R_k'+R_k'',
\end{align*}
where $T_{i}=(S_{\theta_i}-S_{\theta_{i-1}})T_\phi$ and $S_{\theta_i}$ denotes the smoothing operator (see Definition \ref{smdf} for details). As compared with \eqref{kamqk}, our reminder consists of two parts. Namely,  the $R_k''$ is the usual ``square error'', while $R_k'$ represents  a new ``substitution error'' which comes from a further smoothing procedure when we try to solve the homological equations. It is actually the $R_k'$ that dominates the convergence of the scheme:  The total reminder in our scheme  is  now an exponentially small one, rather than a  sup-exponential one as in \eqref{kamqk}. In the iteration scheme, another key ingredient called the tame  property (of the Sobolev norm) plays an essential role. Roughly speaking, the tame  property means in the present that the $s$-norm of the product $\prod\limits_{i=1}^n{A_i}$ is bounded above by a \textit{linear} function of $\|A_i\|_s$ ($i=1,\cdots, n$).

{Definitely,  our proof is based on the natural KAM reducibility idea originated from \cite{Cra83,BLS83,Pos83}.  However,  in the reducibility procedure we  introduce  a Nash-Moser  iteration type  \cite{Nas56, Mos61}  argument  within  the spirit of the scheme introduced by H\"ormander \cite{Hor76, Hor77} (see also \cite{AG07} for an excellent exposition).  This would be explained as follows. First, our proof  employs  some smoothing operators (acting on both  $T_\phi$ and the homological equations) to supply the loss of regularities. Second, we directly perform the iterations  in  the Sobolev norms, and the tame property becomes significantly essential in our proof.   In contrast,  to handle the differentiable Hamiltonians, there is a different approach which is based on approximating differentiable functions by analytic ones (i.e., the Moser-Jackson-Zehnder argument, cf. \cite{Sal04}) known in the KAM field.  The  iteration scheme of the later method uses the analytic norms  as in the standard analytic KAM  one.}

The above argument is a very schematic,  only formal, overview of the proof. It is our main purpose here to actually \textit{develop the whole iteration scheme from scratch}. We also want to point out that we can prove a \textit{quantitative (probably sharp) lower bound} on the regularity required by $T_\phi$, which is given by
\begin{align*}
 \|T_\phi\|_s<\infty\ {\rm with}\ s>\tau+d.
\end{align*}
Similar regularity bound has been previously encountered  in   
localization theory of random models (cf. \cite{AM93}). Finally, while we can't handle the general analytic quasi-periodic potentials,  we believe the method may be improved to resolve at least some special cases, such as the almost Mathieu one (i.e., $d_\mathbf{i}=\cos2\pi(\mathbf{i}\cdot\vec\omega)$).

Once the above reducibility is obtained via the Nash-Moser type iteration, we can apply it to some interesting examples of \cite{Cra83, Sar82, Pos83, BLS83} and establish the (inverse) power-law localization.

\subsection{Structure of paper}

{The structure of paper is then as follows. 

The main results are introduced in \S 2.    In \S 2 our new results  are stated in  theorems and corollaries.  In this section we mainly extend the results of \cite{Cra83,Pos83,BLS83} to the power-law long-range hopping case. So to prove these corollaries, we  need  some  results obtained in \cite{Cra83,BLS83,Pos83},  which are all  cited as lemmas (cf. Lemma \ref{BLS} and Lemma \ref{Pos}).   

Some preliminaries including the tame inequality  (cf. Lemma \ref{tame}) and  properties  of  smoothing operator (cf. Lemma \ref{smooth})  are collected in \S 3. Both Lemma \ref{tame} and Lemma \ref{smooth}  which are essential for our analysis, did  not appear in \cite{Cra83,Pos83,BLS83}. 

The main part of the paper is \S 4, where a Nash-Moser iteration scheme  (cf. Proposition \ref{itprop} and Proposition \ref{prop0}) is stated and proved.  In the proof, we need  a couple of  lemmas,   which  are  all designed for power-law hopping operators and  proved in detail.

 A complete Nash-Moser iteration type  reducibility theorem (cf. Theorem \ref{itthm}) in our setting is stated in  \S5.  
 
 The proofs of our main theorems are contained in \S6.  To prove the main theorems, it suffices to establish the convergence of transformations in Theorem \ref{itthm}. We prove  Theorem \ref{mthm1}  in detail.  We only  outline the proof of  Theorem \ref{mthm2}  since its proof  is just a variant of the preceding one.
 
 A proof of the  tame  inequality  is  given in the appendix.}

\section{Main results}
We  present our results in the language of translation invariant Banach algebra as in P\"oschel \cite{Pos83}.

\subsection{Translation invariant Banach algebra}
Our start point is a Banach algebra $(\mathfrak{B}, \|\cdot\|_\mathfrak{B})$ of complex $d$-dimensional sequences $(a_\mathbf{i})_{\mathbf{i}\in\Z^d}$. The operations are pointwise addition and multiplication of sequences. We assume the sequence $(1)_{\mathbf{i}\in\Z^d}$ belongs to $\mathfrak{B}$ and has norm $1$. Denote by $\sigma_\mathbf{j}, \mathbf{j}\in\Z^d$ the translations on $\mathfrak{B}$, which are defined by  $(\sigma_\mathbf{j}a)_\mathbf{i}=a_{\mathbf{i}-\mathbf{j}}$ for each $a\in \mathfrak{B}$. We assume further that $\mathfrak{B}$ is \textit{translation invariant}, i.e., $\|\sigma_\mathbf{j}a\|_\mathfrak{B}=\|a\|_\mathfrak{B}$ for all $\mathbf{j}\in\Z^d$ and $a\in \mathfrak{B}$.

 It was proved in  \cite{Pos83} that the translation invariance implies for every $a\in \mathfrak{B},$
 \begin{align}\label{infb}
 \|a\|_\infty=\sup_{\mathbf{i}\in\Z^d}|a_\mathbf{i}|\leq\|a\|_\mathfrak{B}.
 \end{align}
 \subsection{The $(\tau,\gamma)$-distal sequence}
  Let $p=(p_\mathbf{i})_{\mathbf{i}\in\Z^d}$ be an arbitrary complex sequence. Fix $\tau>0,\gamma>0$. We call $p$ a $(\tau,\gamma)$-\textit{distal sequence} for $\mathfrak{B}$ if  for $\forall\ \mathbf{k}\in\Z^d\setminus\{\mathbf{0}\}$,
 \begin{align*}
(p-\sigma_\mathbf{k} p)^{-1} \in \mathfrak{B},\ \|(p-\sigma_\mathbf{k} p)^{-1}\|_\mathfrak{B}\leq {\gamma^{-1}}{|\mathbf{k}|^\tau}\
 \end{align*}
 with  $(p-\sigma_\mathbf{k} p)^{-1}_\mathbf{i}=\frac{1}{p_\mathbf{i}-p_{\mathbf{i}-\mathbf{k}}}$ and $|\mathbf{k}|=\sup\limits_{1\leq v\leq d}|k_v|$. We remark that $p$ itself is not necessarily in $\mathfrak{B}.$

 Denote by ${DC}_\mathfrak{B}(\tau,\gamma)$ the set of  all $(\tau,\gamma)$-distal sequence for $\mathfrak{B}$.

 \subsection{The Sobolev Toeplitz operator}
  Now we define the Toeplitz operator with polynomial off-diagonal decay. Let $\mathcal{M}$ be the set of all infinite matrices $A=(a_{\mathbf{i},\mathbf{j}})_{\mathbf{i},\mathbf{j}\in\Z^d}$ satisfying for every $\mathbf{k}\in\Z^d,$
  \begin{align*}
 A_\mathbf{k}=(a_{\mathbf{i},\mathbf{i}-\mathbf{k}})_{\mathbf{i}\in\Z^d}\in \mathfrak{B},
  \end{align*}
 where $A_\mathbf{k}$ is called a $\mathbf{k}$-diagonal of $A$. We can then define the Toeplitz operator of  the Sobolev type. For any $s\in[0,+\infty)$ and $A\in \mathcal{M}$,  define
 \begin{align}\label{sob}
 \|A\|_s^2=\sum_{\mathbf{k}\in\Z^d}\|A_\mathbf{k}\|_\mathfrak{B}^2\langle \mathbf{k}\rangle^{2s},\ \langle \mathbf{k}\rangle=\max\{1, |\mathbf{k}|\}.
 \end{align}
 Then let
  \begin{align*}
 \mathcal{M}^s=\{A\in \mathcal{M}:\ \|A\|_s<\infty\},\ \mathcal{M}^{\infty}=\bigcap_{s\geq 0}\mathcal{M}^s.
 \end{align*}

 For any $A=(a_{\mathbf{i},\mathbf{j}})_{\mathbf{i},\mathbf{j}\in\Z^d}$,  let $\overline{A}={\rm diag}_{\mathbf{i}\in\Z^d}(a_{\mathbf{i},\mathbf{i}})$ be the main diagonal part of $A$.
 Denote by $\mathcal{M}_0^\infty=\{A\in \mathcal{M}^\infty:\  A=\overline{A}\}$ the subspace of all diagonal operators.

  \subsection{The main results}

Let $D={\rm diag}_{\mathbf{i}\in\Z^d}(d_\mathbf{i})$. Denote $D\in DC_\mathfrak{B}(\tau,\gamma)$ if $(d_\mathbf{i})_{\mathbf{i}\in\Z^d}\in DC_\mathfrak{B}(\tau, \gamma).$

 Fix $I\in \mathcal{M}$ to be the identity operator. The first main result is
\begin{thm}[]\label{mthm1}
Fix  $\delta>0, \alpha_0>d/2, \tau>0$ and $ \gamma>0$. Let
\begin{align*}
D\in DC_\mathfrak{B}(\tau, \gamma),\  T\in \mathcal{M}^{\alpha+4\delta},\ \alpha>0.
\end{align*}
Then for
\begin{align*}
\alpha>\alpha_0+\tau+7\delta,
\end{align*}
there exists some $\epsilon_0=\epsilon_0(\delta,\tau,\gamma,\alpha_0,\alpha)>0$ such that the following holds true.  If $\|T\|_{\alpha+4\delta}\leq \epsilon_0$,  then there exist invertible $Q_+\in \mathcal{M}^{\alpha-\tau-7\delta}$  and some  $D_+\in \mathcal{M}_0^\infty$ such that
\begin{align*}
Q_+^{-1}( T+D+D_+)Q_+=D
\end{align*}
with
\begin{align*}
\|Q_+-I\|_{\alpha-\tau-7\delta}&\leq C\|T\|_{\alpha+4\delta}^{\frac{\delta}{\alpha-\alpha_0}},\\
 \|Q_+^{-1}-I\|_{\alpha-\tau-7\delta}&\leq C\|T\|_{\alpha+4\delta}^{\frac{\delta}{\alpha-\alpha_0}},\\
 \|D_+\|_0&\leq C\|T\|_{\alpha+4\delta}^{\frac{1}{\alpha-\alpha_0}},
\end{align*}
where $C=C(\delta,\tau, \alpha_0, \alpha)>0$. Moreover, if both $T$ and $D$ are real symmetric, then
\begin{align*}
U^{-1}(T+D+D_+)U&=D,\\
\|U-I\|_{\alpha-\tau-7\delta}&\leq C\|T\|_{\alpha+4\delta}^{\frac{\delta}{\alpha-\alpha_0}},
\end{align*}
where $U=Q_+(Q_+^tQ_+)^{-\frac12}$ is a unitary operator with $Q_+^t$ denoting the transpose of $Q_+$.

\end{thm}
\begin{rem}
We would like to remark that we have obtained a quantitative (probably sharp) lower bound on the regularity of $T$. Namely,
\begin{align}\label{sharpt}
T\in \mathcal{M}^s\  {\rm with}\  s>d/2+\tau.
\end{align}
 In fact, we have that $\delta>0$ can be arbitrary. If $s>d/2+\tau$, we let $\zeta=s-d/2-\tau>0, \ \alpha=d/2+\tau+\zeta/10,\  \delta=\zeta/1000$ and $ \alpha_0=d/2+\zeta/100$. Then it is easy to check that all the conditions of the theorem are satisfied in this case. Indeed, the lower bound $d/2+\tau$ can be explained as follows, i.e.,   $d/2$ gives the tame information, and $\tau$ the loss of derivatives.
\end{rem}

The second main result is
\begin{thm}[]\label{mthm2}
Fix  $\delta>0, \alpha_0>d/2, \tau>0$ and $\gamma>0$. Let
\begin{align*}
T\in \mathcal{M}^{\alpha+4\delta},\ \alpha>0
\end{align*}
 Assume $D+D'\in DC_{\mathfrak{B}}(\tau,\gamma)$ for any $D'\in \mathcal{M}_0^\infty$ with $\|D'\|_0\leq \eta$ ($\eta>0$).
Then for
\begin{align*}
\alpha>\alpha_0+\tau+7\delta,
\end{align*}
 there exists some $\epsilon_0=\epsilon_0(\eta, \delta,\tau,\gamma,\alpha_0,\alpha)>0$ such that the following holds true.  If $\|T\|_{\alpha+4\delta}\leq \epsilon_0$, then there exist  invertible $Q_+\in \mathcal{M}^{\alpha-\tau-7\delta}$  and some  $D_+\in \mathcal{M}_0^\infty$ such that
\begin{align*}
Q_+^{-1}(T+D)Q_+=D+D_+
\end{align*}
with
\begin{align*}
\|Q_+-I\|_{\alpha-\tau-7\delta}&\leq C\|T\|_{\alpha+4\delta}^{\frac{\delta}{\alpha-\alpha_0}},\\
 \|Q_+^{-1}-I\|_{\alpha-\tau-7\delta}&\leq C\|T\|_{\alpha+4\delta}^{\frac{\delta}{\alpha-\alpha_0}},\\
 \|D_+\|_0&\leq C\|T\|_{\alpha+4\delta}^{\frac{1}{\alpha-\alpha_0}},
\end{align*}
where $C=C(\delta,\tau, \alpha_0, \alpha)>0$. Moreover, if both $T$ and $D$ are real symmetric, then
\begin{align*}
U^{-1}(T+D)U&=D+D_+,\\
\|U-I\|_{\alpha-\tau-7\delta}&\leq C\|T\|_{\alpha+4\delta}^{\frac{\delta}{\alpha-\alpha_0}},
\end{align*}
where $U=Q_+(Q_+^tQ_+)^{-\frac12}$ is a unitary operator.
\end{thm}

\begin{rem}
This theorem shows that if the non-resonant condition is stable under small perturbations, then we can treat the direct problem and obtain the reducibility  of the original operator $T +D$.
\end{rem}

\subsection{Power-law localization}

Now we can apply the above theorems to  obtain the (inverse) \textit{power-law} localization.

Consider the \textit{polynomial} long-range hopping operator 
\begin{align}
\label{tphi}(T_\phi u)_\mathbf{i}&=\sum_{\mathbf{j}\in\Z^d}{\phi_{\mathbf{i}-\mathbf{j}}u_\mathbf{j}},\ \phi=(\phi_\mathbf{i})_{\mathbf{i}\in\Z^d}\in \R^{\Z^d},\\
\label{phi}|\phi_\mathbf{i}|&\leq |\mathbf{i}|^{-s}\ {\rm for}\ \mathbf{i}\in\Z^d\setminus \{\mathbf{0}\},\  \phi_\mathbf{0}=0.
\end{align}
Obviously, $T_\phi\in \mathcal{M}^{s'}$ for $0\leq s'<s-d/2$. 

The first application is
\begin{cor}[]\label{mthm3}
Fix $\alpha_0>d/2, \tau>0, \gamma>0$ and $\delta>0$. Assume that $T_\phi$ is given by \eqref{tphi} and \eqref{phi} with
\begin{align}\label{slbd}
s>\alpha_0+\tau+d/2+12\delta.
\end{align}
Let $D={\rm diag}_{\mathbf{i}\in\Z^d}(d_\mathbf{i})\in DC_\mathfrak{B}(\tau, \gamma)$.
Then there is $\epsilon_0=\epsilon_0(s, d, \delta, \tau, \gamma, \alpha_0)>0$ so that for $\|T_\phi\|_{{s-d/2-\delta}}\leq \epsilon_0$,  the following holds true.  There exists some $D'={\rm diag}_{\mathbf{i}\in\Z^d}(d_\mathbf{i}')$ with
\begin{align*}
D'-D\in \mathcal{M}_0^\infty,\ \|D-D'\|_0\leq C\| T_\phi\|_{{s-d/2-\delta}}^{\frac{1}{s-d/2-5\delta-\alpha_0}}
\end{align*}
such that
\begin{align*}
H'=T_\phi+D'
\end{align*}
has pure point spectrum with a complete set of polynomially localized eigenfunctions $\{e_\mathbf{k}\}_{\mathbf{k}\in\Z^d}$ satisfying
\begin{align*}
|(e_\mathbf{k})_\mathbf{i}|\leq 2\langle \mathbf{i}-\mathbf{k}\rangle^{-s+\tau+d/2+12\delta}\ {\rm for}\ \mathbf{i},  \mathbf{k}\in\Z^d.
\end{align*}
Moreover, if $(d_\mathbf{i})_{\mathbf{i}\in\Z^d}$ is a real-valued sequence, then the spectrum of $H'$ is equal to  the closure of the  sequence $(d_\mathbf{i})_{\mathbf{i}\in\Z^d}$ in $\R$.
\end{cor}
\begin{rem}
Recalling  \eqref{sharpt} and \eqref{phi}, we also have a quantitative lower bound on the regularity of $\phi$, i.e.,
\begin{align*}
|\phi_\mathbf{i}|\leq |\mathbf{i}|^{-s}\ {\rm with}\ {s}>d+\tau.
\end{align*}
This is interesting if one is familiar with the random power-law localization. Consider now $
H_\omega=\varepsilon T_\phi+d_\mathbf{i}(\omega)\delta_{\mathbf{i}\mathbf{i}'},$
where $(d_\mathbf{i}(\omega))_{\mathbf{i}\in\Z^d}$ is a sequence of \textit{i.i.d.} random variables having uniform distribution  on $[0,1]$. Then by developing the remarkable fractional moment method, Aizenman-Molchanov   \cite{AM93} showed  that  for
\begin{align*}
|\phi_\mathbf{i}|\leq |\mathbf{i}|^{-s}\ {\rm with}\ {s}>d,
\end{align*}
 $H_\omega$ exhibits power-law localization for small $\varepsilon$ and a.e. $\omega$.
\end{rem}

\begin{proof}
We apply Theorem \ref{mthm1} with
\begin{align*}
\alpha=s-d/2-5\delta.
\end{align*}
Since \eqref{slbd} and \eqref{phi}, we have
\begin{align*}
&T_\phi\in \mathcal{M}^{s-d/2-\delta}=\mathcal{M}^{\alpha+4\delta},\\
&\alpha=s-d/2-5\delta>\alpha_0+\tau+7\delta.
\end{align*}
Hence using Theorem \ref{mthm1} implies that if $\|T_\phi\|_{s-d/2-\delta}\ll1,$ there are $Q_+\in \mathcal{M},\  D_+\in \mathcal{M}_0^\infty$ so that
\begin{align}
\nonumber&Q_+^{-1}(T_\phi+D+D_+)Q_{+}=D
\end{align}
and
\begin{align}
\label{thm3-1}&\|D_+\|_0\leq  C\|T_\phi\|_{s-d/2-\delta}^{\frac{1}{s-d/2-5\delta-\alpha_0}},\\
\label{thm3-2}&\|Q_+-I\|_{s-d/2-\tau-12\delta}\leq C\| T_\phi\|_{s-d/2-\delta}^{\frac{\delta}{s-d/2-5\delta-\alpha_0}}.
\end{align}
We let $D'=D+D_+$. Then the estimate on $D'-D$ follows from  \eqref{thm3-1}.

Next, note that $D$ is  a diagonal operator. Then the standard basis $\{\delta_\mathbf{k}\}_{\mathbf{k}\in\Z^d}$ of $\ell^2(\Z^d)$  is a complete set of eigenfunctions of $D$ with eigenvalues $\{d_\mathbf{k}\}_{\mathbf{k}\in\Z^d}$.
Then
\begin{align*}
H'Q_+\delta_\mathbf{k}=Q_+D\delta_\mathbf{k}=d_\mathbf{k}Q_{+}\delta_\mathbf{k},
\end{align*}
which implies $\{e_\mathbf{k}\}_{\mathbf{k}\in\Z^d}$ $( e_\mathbf{k}=Q_+\delta_\mathbf{k})$ may be  a \textit{complete set of eigenfunctions} of $H'$. Indeed, we obtain since  \eqref{thm3-2} that
\begin{align*}
\|Q_+\|_{s-d/2-\tau-12\delta}\leq  2 \ {\rm if}\  \|T_\phi\|_{s-d/2-\delta}\ll1,
\end{align*}
which together with \eqref{infb} implies
\begin{align*}
|(Q_+\delta_\mathbf{k})_\mathbf{i}|&=|(Q_+)_{\mathbf{i},\mathbf{k}}|=|(Q_+)_{\mathbf{i}-\mathbf{k}}(\mathbf{i})|\leq \|(Q_+)_{\mathbf{i}-\mathbf{k}}\|_\infty\\
&\leq \|(Q_+)_{\mathbf{i}-\mathbf{k}}\|_\mathfrak{B}\leq \|Q_+\|_{s-d/2-\tau-12\delta}\langle\mathbf{ i}-\mathbf{k}\rangle^{-s+d/2+\tau+12\delta}\\
&\leq 2\langle \mathbf{i}-\mathbf{k}\rangle^{-s+d/2+\tau+12\delta},
\end{align*}
where $(Q_{+})_{\mathbf{i}-\mathbf{k}}$ denotes the $\mathbf{i}-\mathbf{k}$ diagonal of $Q_+.$ Since
$$s-d/2-\tau-12\delta>\alpha_0>d/2,$$
we have $e_\mathbf{k}\in\ell^2(\Z^d)$. This shows that $\{e_\mathbf{k}\}_{\mathbf{k}\in\Z^d}$ are eigenfunctions of $H'$. We then show the \textit{completeness}. Denote by $(\cdot, \cdot)$ the standard inner product on $\ell^2(\Z^d).$ Let $\varphi\in\ell^2(\Z^d)$ and suppose for all $\mathbf{k}\in\Z^d$, $(\varphi, e_\mathbf{k})=0$. It suffices to show $\varphi=0$. Assume $\varphi\neq 0.$ Choose a $\mathbf{k}\in\Z^d$ so that
$|\varphi_\mathbf{k'}|\leq |\varphi_\mathbf{k}|$ for all $\mathbf{k'}\in\Z^d.$  Then $\varphi_\mathbf{k}\neq0$
However,  we have
\begin{align*}
0=(\varphi, e_\mathbf{k})=\varphi_\mathbf{k}+(\varphi, (Q_+-I)\delta_\mathbf{k}),
\end{align*}
which together with \eqref{thm3-2} implies
\begin{align*}
|\varphi_\mathbf{k}|=|(\varphi, (Q_+-I)\delta_\mathbf{k})|
&\leq |\varphi_\mathbf{k}|\sum_{\mathbf{j}\in\Z^d}|(Q_+-I)_{\mathbf{j},\mathbf{k}}|\\
&\leq |\varphi_\mathbf{k}|\sum_{\mathbf{j}\in\Z^d}\|(Q_+-I)_{\mathbf{j}-\mathbf{k}}\|_\mathfrak{B}\\
&\leq \|Q_+-I\|_{\alpha_0}|\varphi_\mathbf{k}|\sum_{\mathbf{j}\in\Z^d}{\langle \mathbf{j}-\mathbf{k}\rangle}^{-\alpha_0}\\
 &\leq C \|T_\phi\|_{s-d/2-\delta}^{\frac{\delta}{s-d/2-5\delta-\alpha_0}}|\varphi_\mathbf{k}|=o(1)|\varphi_\mathbf{k}|.
\end{align*}
This contradicts $\varphi_\mathbf{k}\neq 0.$ We have proven the \textit{completeness}.

Finally, if $(d_\mathbf{i})_{\mathbf{i}\in\Z^d}\in\R^{\Z^d}$,  then $Q_+$ can be replaced with a unitary operator. Consequently, the spectrum of $H'$ is equal to that of $D.$

This proves Corollary \ref{mthm3}.
\end{proof}

\begin{cor}[]\label{mthm4}
Fix $\alpha_0>d/2, \tau>0, \gamma>0$ and $\delta>0$. Assume that $T_\phi$ is given by \eqref{tphi} and \eqref{phi} with
\begin{align*}
s>\alpha_0+\tau+d/2+12\delta.
\end{align*}
Let $D={\rm diag}_{\mathbf{i}\in\Z^d}(d_\mathbf{i})\in DC_\mathfrak{B}(\tau, \gamma)$.   Assume further $(p_\mathbf{i}+d_\mathbf{i})_{\mathbf{i}\in\Z^d}\in DC_\mathfrak{B}(\tau, \gamma)$  for each $ P={\rm diag}_{\mathbf{i}\in\Z^d}(p_\mathbf{i})\in \mathcal{M}_0^\infty$ satisfying $\|P\|_0\leq \eta,\ \eta>0.$
Then there is $\epsilon_0=\epsilon_0(\eta, s, d, \delta, \tau, \gamma, \alpha_0)>0$ so that for $\|T_\phi\|_{{s-d/2-\delta}}\leq \epsilon_0$,  the following holds true.  There exists some $D'={\rm diag}_{\mathbf{i}\in\Z^d}(d_\mathbf{i}')$ with
\begin{align*}
D'-D\in \mathcal{M}_0^\infty,\ \|D-D'\|_0\leq C\|T_\phi\|_{{s-d/2-\delta}}^{\frac{1}{s-d/2-5\delta-\alpha_0}}
\end{align*}
such that
\begin{align*}
H=T_\phi+D
\end{align*}
has pure point spectrum with a complete set of polynomially localized eigenfunctions $\{e_\mathbf{k}\}_{\mathbf{k}\in\Z^d}$ satisfying
\begin{align*}
|(e_\mathbf{k})_\mathbf{i}|\leq 2\langle \mathbf{i}-\mathbf{k}\rangle^{-s+\tau+d/2+12\delta}\ {\rm for}\ \mathbf{i}, \mathbf{k}\in\Z^d.
\end{align*}
Moreover, if $(d_\mathbf{i})_{\mathbf{i}\in\Z^d}$ is a real-valued sequence, then the spectrum of $H$ is equal to the closure of the  sequence $(d_\mathbf{i}')_{\mathbf{i}\in\Z^d}$ in $\R$.
\end{cor}
\begin{proof}
Recalling Theorem \ref{mthm2}, the proof is similar to that of Corollary \ref{mthm3}. We omit the details here.
\end{proof}
\subsection{Application to almost-periodic operators}
In this section we apply our results to some concrete examples of  almost-periodic operators and prove corresponding (inverse) power-law localization.
\begin{ex}
 We  first revisit P\"oschel's \cite{Pos83} limit-periodic potentials. Let $\mathcal{P}$ denote the set of all $(a_\mathbf{i})_{\mathbf{i}\in\Z^d}$ with period $2^n, n\geq0,$ i.e., $a_\mathbf{i}=a_\mathbf{j}\ {\rm if}\ \mathbf{i}-\mathbf{j}\in 2^n\Z^d$.
Define $\mathfrak{L}$ to be the closure of $\mathcal{P}$ in the $\|\cdot\|_\infty$ norm. Then $\mathfrak{L}$ is a translation invariant Banach algebra. In the following we generalize P\"oschel's results to the polynomial long-range hopping case: There exists some polynomial long-range hopping limit-periodic operator of which the spectrum is pure point and can be either $[0,1]$, or the standard Cantor set. Precisely, we have
\begin{cor}[]
Fix $\alpha_0>d/2$ and $\delta>0$. Assume that $T_\phi$ is given by \eqref{tphi} and \eqref{phi} with
\begin{align*}
s>\alpha_0+(\log_23+1/2)d+12\delta.
\end{align*}
Then there is $\epsilon_0=\epsilon_0(s, d, \delta,  \alpha_0)>0$ so that for $\|T_\phi\|_{{s-d/2-\delta}}\leq \epsilon_0$,  the following holds true.  There exists some $d'\in \mathfrak{L}$
such that
\begin{align*}
H'=T_\phi+d'_\mathbf{i}\delta_{\mathbf{ii'}}
\end{align*}
has pure point spectrum with a complete set of polynomially localized eigenfunctions $\{e_\mathbf{k}\}_{\mathbf{k}\in\Z^d}$ satisfying
\begin{align*}
|(e_\mathbf{k})_\mathbf{i}|\leq 2\langle \mathbf{i}-\mathbf{k}\rangle^{-s+(\log_23+1/2)d+12\delta}\ {\rm for}\ \mathbf{i}, \mathbf{k}\in\Z^d.
\end{align*}
Moreover,  the spectrum  of $H'$ can be either  $[0,1]$, or be a Cantor set of zero Lebesgue measure.
\end{cor}
\begin{rem}
We refer to \cite{Avi09,DG10,DG11,DG13} for more results on Schr\"odinger operators with limit-periodic potentials.
\end{rem}

\begin{proof}
It needs to construct limit-periodic sequences and then apply Corollary \ref{mthm3}.

Such sequences were introduced by P\"oschel \cite{Pos83}.
First, for any $v>0$, define characteristic function $\chi_{A_v}:\ \Z\to\{0,1\},$ where $A_v=\bigcup_{n\in\Z}[n2^v, n2^v+2^{v-1})$ for even $v$, and $A_v=\bigcup_{n\in\Z}[n2^v+2^{v-1}, n2^v+2^{v})$ for odd $v$. Then the desired sequences $d''=(d''_\mathbf{i})_{\mathbf{i}\in\Z^d}$ and $d'''=(d'''_\mathbf{i})_{\mathbf{i}\in\Z^d}$  are defined by
\begin{align*}
d_\mathbf{i}''=\sum_{v=1}^\infty\sum_{u=1}^d\chi_{A_v}({i}_u)2^{-(v-1)d-u}, \ \mathbf{i}=(i_1,\cdots, i_d)\in\Z^d,\\
d_\mathbf{i}'''=2\sum_{v=1}^\infty\sum_{u=1}^d\chi_{A_v}(i_u)3^{-(v-1)d-u}, \ \mathbf{i}=(i_1,\cdots, i_d)\in\Z^d.
\end{align*}
P\"oschel (see \textbf{Example} \textbf{1} and \textbf{2} in \cite{Pos83}) has proven that $d''\in \mathfrak{L}$ is a $(d, 16^{-d})$-distal sequence for $\mathfrak{L}$, and $d'''$ is a $(d\log_23, 3^{-d})$-distal one. Moreover, $\{d_\mathbf{i}''\}$ is dense in $[0,1]$, and $\{d_\mathbf{i}'''\}$ is dense in the standard Cantor set of real numbers in $[0,1].$

Finally, it suffices to apply Corollary \ref{mthm3} with $\tau=d\log_23, \gamma=16^{-d}$ and $D={\rm diag}_{\mathbf{i}\in\Z^d}(d_\mathbf{i}'')$ or $D={\rm diag}_{\mathbf{i}\in\Z^d}(d_\mathbf{i}''')$.
\end{proof}

\end{ex}

In the following we apply Corollary \ref{mthm3} and \ref{mthm4} to the quasi-periodic operators case. For $x\in\R, \vec{\omega}=(\omega_1,\cdots, \omega_d)\in\R^d$ and $ \mathbf{i}\in\Z^d$,  define
\begin{align*}
\|x\|_{\R/\Z}=\inf_{k\in\Z}|x-k|,\ \mathbf{i}\cdot\vec\omega=\sum_{v=1}^di_v\omega_v.
\end{align*}

We have

\begin{cor}[]\label{mthm5}
Fix $\alpha_0>d/2, \tau>0, \gamma>0$ and $\delta>0$. Assume that $T_\phi$ is given by \eqref{tphi} and \eqref{phi} with
\begin{align*}
s>\alpha_0+\tau+d/2+12\delta.
\end{align*}
Let $(\mathfrak{F}, \|\cdot\|_{\mathfrak{F}})$ be a translation invariant Banach algebra of functions defined on $\C$, and let $\vec\omega\in\R^d.$ If $f$ is an arbitrary function satisfying
\begin{align}\label{DC}
(f-\sigma_{\mathbf{i}\cdot\vec\omega}f)^{-1}\in \mathfrak{F}\ {\rm and}\  \|(f-\sigma_{\mathbf{i}\cdot\vec\omega}f)^{-1}\|_\mathfrak{F}\leq {\gamma^{-1}}{|\mathbf{i}|^\tau}\ {\rm for}\ \forall\ \mathbf{i}\in\Z^d\setminus\{\mathbf{0}\}
\end{align}
with $(\sigma_{\mathbf{i}\cdot\vec\omega}f)(x)=f(x-\mathbf{i}\cdot\vec\omega),$ then there is some $\epsilon_0=\epsilon_0(s, d, \delta, \tau, \gamma, \alpha_0)>0$ so that for $\|T_\phi\|_{{s-d/2-\delta}}\leq \epsilon_0$,  the following holds true.  There exists a function $f'$ with $\|f'-f\|_\mathfrak{F}\leq C\|T_\phi\|_{{s-d/2-\delta}}^{\frac{1}{s-d/2-5\delta-\alpha_0}}$
such that
\begin{align*}
H'=T_\phi+f'(\mathbf{i}\cdot\vec\omega)\delta_{\mathbf{ii}'}
\end{align*}
has pure point spectrum with a complete set of polynomially localized eigenfunctions $\{e_\mathbf{k}\}_{\mathbf{k}\in\Z^d}$ satisfying
\begin{align*}
|(e_\mathbf{k})_\mathbf{i}|\leq 2\langle \mathbf{i}-\mathbf{k}\rangle^{-s+\tau+d/2+12\delta}\ {\rm for}\ \mathbf{i}, \mathbf{k}\in\Z^d.
\end{align*}
Moreover, if $f$ is a real-valued function,  then  spectrum of $H'$ is equal to the closure of the  sequence $(f(\mathbf{i}\cdot\vec\omega))_{\mathbf{i}\in\Z^d}$ in $\R$.

If the condition \eqref{DC} holds for all $f'$ with $\|f'-f\|_\mathfrak{F}\leq \eta$, then the above statement also holds true with the role of $f, f'$ interchanged.
\end{cor}
\begin{proof}
As done in P\"oschel \cite{Pos83}, we first define sequence Banach algebra $(\mathfrak{B}, \|\cdot\|_\mathfrak{B})$ via the following: Let $\mathfrak{B}$ be the set of all sequence $(a_\mathbf{i})_{\mathbf{i}\in\Z^d}$ satisfying $a_\mathbf{i}=f(\mathbf{i}\cdot\vec\omega)$ for some $f\in \mathfrak{F}$. Define $\|a\|_\mathfrak{B}=\inf_{f\in \mathfrak{F},\ f(\mathbf{i}\cdot\vec\omega)=a_\mathbf{i}}\|f\|_\mathfrak{F}$. P\"oschel has proven that $(\mathfrak{B}, \|\cdot\|_\mathfrak{B})$ is a translation invariant Banach algebra. Thus to prove the corollary, it suffices to use Corollary \ref{mthm3} and \ref{mthm4} in this setting.
\end{proof}

In the following we assume $\vec\omega$ satisfies the Diophantine condition
\begin{align}\label{DC1}
\|\mathbf{i}\cdot\vec\omega\|_{\R/\Z}\geq \frac{\gamma}{|\mathbf{i}|^\tau}\ {\rm for}\ \forall\ \mathbf{i}\in\Z^d\setminus\{\mathbf{0}\}.
\end{align}

\begin{ex}[Meromorphic  potential, \cite{BLS83, Sar82}]
This type of potentials was first introduced by \cite{BLS83}, which includes the Maryland potential \cite{GFP82} and Sarnak's potential \cite{Sar82} as special cases. More precisely,
if $r>0,$  $\mathfrak{H}_r$ denotes the set of period $1$ holomorphic functions on $\mathcal{D}_r=\{z\in\C:\ |\Im z|<r\}$ satisfying
\begin{align*}
\|f\|_{\mathfrak{H}_r}:=\sup_{z\in \mathcal{D}_r}\left(|f(z)|+\left|\frac{df(z)}{dz}\right|\right)<\infty.
\end{align*}
Then $(\mathfrak{H}_r, \|\cdot\|_{\mathfrak{H}_r})$ becomes a translation invariant Banach algebra (see \cite{Pos83}).

Let $\mathfrak{P}_r$  denote the set of period $1$ meromorphic functions on $\mathcal{D}_r$ satisfying
\begin{align*}
f\in \mathfrak{P}_r\Leftrightarrow\exists\  c=c(r)>0\ s.t.,\ \inf_{x\in\R,\ z\in \mathcal{D}_r}|f(z)-f(z-x)|\geq c\|x\|_{\R/\Z}.
\end{align*}
The following lemma is due to Bellissard-Lima-Scoppola. 
\begin{lem}[\cite{BLS83}]\label{BLS}
Fix $r>0$ and let $\vec\omega$ satisfy \eqref{DC1}. If $f\in \mathfrak{P}_r, 0<r_1<r$, then there are $C>0, \eta>0$ so that for every $g\in \mathfrak{H}_r$ with $\|g\|_{\mathfrak{H}_r}\leq \eta$,  one has $f_1=f+g\in \mathfrak{P}_{r-r_1}$ and
\begin{align*}
(f_1-\sigma_{\mathbf{i}\cdot\vec\omega}f_1)^{-1}\in \mathfrak{H}_{r-r_1},\ \|(f_1-\sigma_{\mathbf{i}\cdot\vec\omega}f_1)^{-1}\|_{\mathfrak{H}_{r-r_1}}\leq C\gamma^{-1}|\mathbf{i}|^\tau\ {\rm for}\ \forall\  {\mathbf{i}\neq\mathbf{0}}.
\end{align*}
Moreover, the functions $\tan(\pi z),\  \exp {(2\pi\sqrt{-1} z)}\in \mathfrak{P}_r$ for any $r>0$.
\end{lem}

In this example we deal with
\begin{align}\label{mary}
H=T_\phi+f(\mathbf{i}\cdot\vec\omega)\delta_{\mathbf{ii}'}, \ f\in\mathfrak{P}_r\ (r>0).
\end{align}

Combining Corollary \ref{mthm5} and Lemma \ref{BLS} yields
\begin{cor}[]
Fix $\alpha_0>d/2, \tau>0, \gamma>0$ and $\delta>0$. Assume that $T_\phi$ is given by \eqref{tphi} and \eqref{phi} with
\begin{align*}
s>\alpha_0+\tau+d/2+12\delta.
\end{align*}
Let $\vec\omega$ satisfy \eqref{DC1}, and let $H$ be defined by \eqref{mary}. Then there is some $\epsilon_0=\epsilon_0(f, r, s, d, \delta, \tau, \gamma, \alpha_0)>0$ so that for $\|T_\phi\|_{{s-d/2-\delta}}\leq \epsilon_0$,   $H$ has pure point spectrum with a complete set of polynomially localized eigenfunctions $\{e_\mathbf{k}\}_{\mathbf{k}\in\Z^d}$ satisfying
\begin{align*}
|(e_\mathbf{k})_\mathbf{i}|\leq 2\langle \mathbf{i}-\mathbf{k}\rangle^{-s+\tau+d/2+12\delta}\ {\rm for}\ \mathbf{i}, \mathbf{k}\in\Z^d.
\end{align*}
In particular, the above statements hold true for
\begin{align*}
H=T_\phi+\tan(\pi \mathbf{i}\cdot\vec\omega)\delta_{\mathbf{ii}'}
\end{align*}
and
\begin{align*}
H=T_\phi+\exp(2\pi \sqrt{-1}\mathbf{i}\cdot\vec\omega)\delta_{\mathbf{ii}'}.
\end{align*}
\end{cor}
\begin{rem}
(1). The Maryland type potential $\tan (\textbf{i}\cdot\vec\omega)$ was first introduced in \cite{GFP82}. It is well-known that the  Maryland type models (even with an exponential long-range hopping) are  solvable ones \cite{GFP82, FP84, Sim85}. We refer to \cite{JL17} for an almost complete description of the spectral types for the standard Maryland model (i.e., the $1D$ Schr\"odinger operator with the Maryland potential). Very recently, a  Rayleigh-Schr\"odinger perturbation type method was   developed by \cite{KPS20} to prove Anderson localization for some quasi-periodic operators with Maryland type potentials.

(2). The Schr\"odinger operator with the potential $\exp(2\pi \sqrt{-1}\mathbf{i}\cdot\vec\omega)$ was first introduced in \cite{Sar82}, and is a typical non-normal operator. While the spectral theorem is not available for the non-normal operator, we can still define pure point spectrum for it: We say an arbitrary operator on $\ell^2(\Z^d)$ is of pure point spectrum if it has a complete set of eigenfunctions.

\end{rem}

\end{ex}

\begin{ex}[Discontinuous potential, \cite{Cra83}]
We then introduce the potential of \cite{Cra83}. Let $\mathfrak{I}$ denote the space of all functions of period $1$ and  bounded variation with the norm
\begin{align*}
\|f\|_\mathfrak{I}=\sup_{x\in\R}|f(x)|+\|f\|_{BV},
\end{align*}
where $\|\cdot\|_{BV}$ denotes the standard total variation. Then $(\mathfrak{I}, \|\cdot\|_\mathfrak{I})$ is a translation invariant Banach algebra.

Obviously,  a typical example of such a potential is $f(x)=x\mod 1$ which has been previously studied by  Craig \cite{Cra83}.

In this example  we study
\begin{align*}
 T_\phi+((\mathbf{i}\cdot\vec\omega)\mod 1)\delta_{\mathbf{ii}'}.
\end{align*}
The following result  can be found in  \cite{Pos83}.
\begin{lem}[\cite{Pos83}]\label{Pos}
Let $\vec\omega$ satisfy \eqref{DC1}, and let $f(x)=x\mod 1$. Then
\begin{align*}
(f-\sigma_{\mathbf{i}\cdot\vec\omega}f)^{-1}\in \mathfrak{I},\ \|(f-\sigma_{\mathbf{i}\cdot\vec\omega}f)^{-1}\|_{\mathfrak{I}}\leq \gamma^{-1} |\mathbf{i}|^\tau\ {\rm for}\ \mathbf{i}\neq\mathbf{0}.
\end{align*}
\end{lem}

Combining Corollary \ref{mthm5} and Lemma \ref{Pos} yields
\begin{cor}[]
Fix $\alpha_0>d/2, \tau>0, \gamma>0$ and $\delta>0$. Assume that $T_\phi$ is given by \eqref{tphi} and \eqref{phi} with
\begin{align*}
s>\alpha_0+\tau+d/2+12\delta.
\end{align*}
Let $\vec\omega$ satisfy \eqref{DC1}. Then there is $\epsilon_0=\epsilon_0(s, d, \delta, \tau, \gamma, \alpha_0)>0$ so that for $\|T_\phi\|_{{s-d/2-\delta}}\leq \epsilon_0$,  the following holds true.  There exists some $d'\in \mathfrak{I}$ with
\begin{align*}
\|(x\mod\ 1)-d'\|_\mathfrak{I}\leq C\|T_\phi\|_{{s-d/2-\delta}}^{\frac{1}{s-d/2-5\delta-\alpha_0}}
\end{align*}
such that
\begin{align*}
H'=T_\phi+d'(\mathbf{i}\cdot\vec\omega)\delta_{\mathbf{ii}'}
\end{align*}
has pure point spectrum with a complete set of polynomially localized eigenfunctions $\{e_\mathbf{k}\}_{\mathbf{k}\in\Z^d}$ satisfying
\begin{align*}
|(e_\mathbf{k})_\mathbf{i}|\leq 2\langle \mathbf{i}-\mathbf{k}\rangle^{-s+\tau+d/2+12\delta}\ {\rm for}\ \mathbf{i}, \mathbf{k}\in\Z^d.
\end{align*}
Moreover,  the set of all eigenvalues  for $H'$ is just  $\{(\mathbf{i}\cdot\vec\omega \mod 1):\ \mathbf{i}\in\Z^d\}$.
\end{cor}
\begin{rem}
A natural generation of the potential $x \mod 1$ is the so called Lipschitz monotone potential, which was introduced recently by \cite{JK19}. In \cite{JK19}  the authors proved  all couplings localization for the $1D$ quasi-periodic Schr\"odinger operator with a Lipschitz monotone potential. 
\end{rem}

\end{ex}



\section{Preliminaries}

\subsection{The notation}
\begin{itemize}
\item Denote by $I$ the identity operator on $\mathcal{M}$.
\item For two operators $X, Y$, we define
\begin{align*}
[X, Y]=XY-YX.
\end{align*}

\item By $C=C(f)>0$, we mean the constant $C$ depends only on $f.$

\item Typically, we write $\mathbf{x}\in\R^d$ for a vector,  and $x\in \R$ for a scalar. If $\mathbf{k}\in\Z^d$, then let
\begin{align*}
|\mathbf{k}|=\max_{1\leq v\leq d}|k_v|,\ \langle \mathbf{k}\rangle=\max\{1, |\mathbf{k}|\}.
\end{align*}

\item By $\delta_{\mathbf{ij}}$ we denote the Kronecker delta. Namely, $\delta_{\mathbf{ii}}=1$ , and $\delta_{\mathbf{ij}}=0$ if $\mathbf{i}\neq \mathbf{j}.$

\item By ${\rm diag}_{\mathbf{i}\in\Z^d}(d_\mathbf{i})$ we denote the diagonal matrix whose diagonal elements are $(d_\mathbf{i})_{\mathbf{i}\in\Z^d}.$
\item For any $A=(a_{\mathbf{i},\mathbf{j}})_{\mathbf{i},\mathbf{j}\in\Z^d}$,  let $\overline{A}={\rm diag}_{\mathbf{i}\in\Z^d}(a_{\mathbf{i},\mathbf{i}})$ be the main diagonal part of $A$.
\item By $0\leq x\ll y$ we mean there is a small $c>0$ so that $x\leq cy.$
\item Throughout the paper $\alpha_0>d/2, \gamma>0$ are fixed.
\end{itemize}

\subsection{Tame property}
The Sobolev norm defined by \eqref{sob} has the following important tame property (which is also called the interpolation property). Such tame property is not shared by the exponential one.
\begin{lem}\label{tame}
Fix $\alpha_0>d/2.$ Then for any $s\geq \alpha_0$ and $X,Y\in \mathcal{M}^{s}$, we have
\begin{align}\label{tame1}
\|XY\|_s\leq K_0\|X\|_{\alpha_0}\|Y\|_{s}+K_1\|X\|_{s}\|Y\|_{\alpha_0},
\end{align}
where
\begin{align}
\label{k0}K_0&=\sqrt{20\sum_{\mathbf{k}\in\Z^d}\langle \mathbf{k}\rangle^{-2\alpha_0}},\\
\label{k1}K_1(s)&=(1-10^{-\frac{1}{2s}})^{-s}\sqrt{2\sum_{\mathbf{k}\in\Z^d}\langle \mathbf{k}\rangle^{-2\alpha_0}}.
\end{align}
\end{lem}
\begin{rem}
Obviously, $K_1(\alpha_0,s)$ is bounded for $s$ is a bounded interval. This Sobolev type norm was previously used by Berti-Bolle \cite{BB13} to prove the existence of finitely smoothing quasi-periodic solutions for some nonlinear Schr\"odinger equations. Recently, such norm was also used by Shi \cite{Shi21} to give a multi-scale proof of power-law localization for some random operators.
\end{rem}
\begin{proof}
For a detailed proof, we refer to the appendix.
\end{proof}

\subsection{Smoothing operator}
The smoothing operator plays  an  essential role in the Nash-Moser iteration scheme. In the present context we have
\begin{defn}\label{smdf}
Fix $\theta>0$. Define the smoothing operator $S_\theta: \mathcal{M}\to \mathcal{M}^\infty$ by
\begin{align*}
(S_\theta X)_{\mathbf{i},\mathbf{j}}&=X_{\mathbf{i},\mathbf{j}}\ {\rm for}\ |\mathbf{i}-\mathbf{j}|\leq \theta,\\
(S_\theta X)_{\mathbf{i},\mathbf{j}}&=0\ {\rm for}\ |\mathbf{i}-\mathbf{j}|> \theta.
\end{align*}
\end{defn}

\begin{lem}\label{smooth}
Fix $\theta>0$. Then  for $X\in \mathcal{M}^{s'}$, we have
\begin{align}
\label{smooth1}\|S_\theta X\|_s&\leq \theta^{s-s'}\|X\|_{s'}\ {\rm for}\ s\geq s'\geq0,\\
\label{smooth2}\|(I-S_\theta) X\|_s&\leq  \theta^{s-s'}\|X\|_{s'}\ {\rm for}\ 0\leq s\leq s'.
\end{align}
\end{lem}
\begin{proof}
Recalling the Definition \ref{smdf} and \eqref{sob}, the proof is  trivial.
\end{proof}


\section{The Nash-Moser scheme}

The main scheme is to find a sequence of invertible operators $Q_0, \cdots,  Q_k, \cdots $ such that
\begin{align*}
Q_k^{-1}\left(\sum_{l=1}^kT_{l-1}+ D+\sum_{l=1}^k D_{l-1}\right)Q_k=D+R_k,
\end{align*}
where
\begin{align*}
T_l=(S_{\theta_{l}}-S_{\theta_{l-1}})T,\ \theta_l=\theta_0 \Theta^l,
\end{align*}
where $\theta_0$ and  $\Theta>1$ will be specified later. We will show $ \sum_{l=1}^k T_{l-1}\to T, Q_k\to Q$, $\sum_{l=1}^kD_{l-1}\to D_+$ and $R_k\to 0$ as $k\to\infty$ in the $s$-norm for some $s>0.$

The key difficulty in the above scheme is the so called \textit{loss of derivatives} when we determine  $Q_k$.  Since we are in the Sobolev case, we have a $\tau$-order loss of derivative at each iteration step.  Such a loss is not presented in the analytic case (the loss in the analytic norm is of order $\delta$  for \textit{arbitrary} $\delta>0$, cf. \cite{Pos83, BLS83}).


\subsection{The iteration step}
We first deal with the iteration step. We will work with $s$-Sobolev norms with
\begin{align*}
s\in[\alpha_0, \alpha_1],
\end{align*}
where $\alpha_0>d/2$ is fixed and $\alpha_1$ will be specified later.

We also assume \footnote{The exponential scale of smoothing operator was first introduced by  Klainerman \cite{Kla80}.}
\begin{align}\label{Tl}
\theta_{l}=\theta_0\Theta^{l},\  T_l=(S_{\theta_l}-S_{\theta_{l-1}})T \ (l\geq 1).
\end{align}
where $\theta_0>1$ and $\Theta>1$ will be specified later.
We have the following iteration proposition.
\begin{prop}\label{itprop}
Fix $\alpha_0>d/2, \tau>0, \gamma>0, \delta>0.$ Assume that
\begin{align*}
\alpha>\tau+\alpha_0+7\delta,\ \alpha_1\geq2\alpha+\delta.
\end{align*}
Assume further that
\begin{align*}
\|T\|_{\alpha+4\delta}\leq 1,\ \Theta\geq \max\{8^{2/\delta}(K_0+K_1(\alpha_0))^{4/\delta}, 10^{1/\delta}, 10^{1/\alpha_0} \}.
\end{align*}

If for $1\leq l\leq k$, there exist  $W_l, R_l, Q_l\in \mathcal{M}$ and $D_{l-1}\in \mathcal{M}_0^\infty$  so that
\begin{align*}
Q_l^{-1}H_lQ_l=D+R_l,\ H_l= \sum_{i=1}^l(T_{i-1}+D_{i-1})+D
\end{align*}
with
\begin{align*}
Q_0&=I, D_0=0, T_0=S_{\theta_0}T,\\
 Q_l&=Q_{l-1} V_l, V_l=I+W_l,\\
\|W_l\|_{s}&\leq \theta_{l-1}^{s-\alpha+\tau+4\delta}\ {\rm for}\ s\in[\alpha_0, \alpha_1],\\
\|V_l^{-1}-I\|_{s}&\leq 2K_1(s)\theta_{l-1}^{s-\alpha+\tau+4\delta}\ {\rm for}\ s\in[\alpha_0, \alpha_1],\\
 \|R_{l}\|_s&\leq \theta_{l}^{s-\alpha}\ {\rm for}\ s\in[\alpha_0, \alpha_1],\\
\|D_{l-1}\|_0&\leq 3\theta_{l-1}^{\alpha_0-\alpha}, 
\end{align*}
 where $ K_1(s), T_l$ are defined by \eqref{k1} and \eqref{Tl} respectively.
 Then there is some $C=C(\delta, \Theta, \tau, \gamma, \alpha, \alpha_0,\alpha_1)>0$ such that,  for $\theta_0\geq C$,  there exist $W_{k+1}, R_{k+1}, Q_{k+1}\in \mathcal{M}$ and  $D_{k}\in \mathcal{M}_0^\infty$ so that
\begin{align}\label{qhq1}
Q_{k+1}^{-1}H_{k+1}Q_{k+1}=D+R_{k+1},\ H_{k+1}=H_{k}+T_k+D_k
\end{align}
with
\begin{align*}
Q_{k+1}&=Q_kV_{k+1},\ V_{k+1}=I+W_{k+1},\\
\|W_{k+1}\|_{s}&\leq \theta_{k}^{s-\alpha+\tau+4\delta}\ {\rm for}\ s\in[\alpha_0, \alpha_1],\\
\|V_{k+1}^{-1}-I\|_{s}&\leq 2K_1(s)\theta_{k}^{s-\alpha+\tau+4\delta}\ {\rm for}\ s\in[\alpha_0, \alpha_1],\\
\|R_{k+1}\|_s&\leq \theta_{k+1}^{s-\alpha}\ {\rm for}\ s\in[\alpha_0, \alpha_1],\\
\|D_{k}\|_{0}&\leq 3\theta_{k}^{\alpha_0-\alpha}.
\end{align*}
\end{prop}

\begin{proof}[\bf Proof of Proposition \ref{itprop}]

We first do some formal computations. Note that for $1\leq l\leq k,$  all the $ W_l, R_l, Q_l, D_{l-1}$ have been constructed so that
\begin{align*}
Q_k^{-1}H_kQ_k=D+R_k.
\end{align*}
We try to determine $W_{k+1}, R_{k+1}, D_{k}$ so that \eqref{qhq1} holds. For this purpose, we set
\begin{align*}
&\ \ \ V_{k+1}^{-1}Q_k^{-1}(H_k+T_k+D_{k})Q_kV_{k+1}\\
&=V_{k+1}^{-1}Q_k^{-1}H_kQ_kV_{k+1} +V_{k+1}^{-1}Q_k^{-1}(T_{k}+D_{k}) Q_kV_{k+1}\\
&=V_{k+1}^{-1}(D+R_k)V_{k+1}+V_{k+1}^{-1}Q_k^{-1}(T_{k}+D_{k}) Q_kV_{k+1}.
\end{align*}
We hope that $V_{k+1}=I+W_{k+1}$ and $W_{k+1}$ would be of order $O(\|R_k\|_s)$.
Then
\begin{align*}
V_{k+1}^{-1}(D+R_k)V_{k+1}
&=D+[D, W_{k+1}]+R_{k}+R_{(1)},
\end{align*}
where
\begin{align}
\nonumber R_{(1)}&=V_{k+1}^{-1}(D+R_k)V_{k+1}-D-[D, W_{k+1}]-R_{k}\\
\nonumber&=V_{k+1}^{-1}DV_{k+1}-D-[D, W_{k+1}]+V_{k+1}^{-1}R_kV_{k+1}-R_k\\
\label{R(1)}&=(V_{k+1}^{-1}-I)[D, W_{k+1}]+(V_{k+1}^{-1}-I)R_kW_{k+1}\\
\nonumber&\ \ \ +(V_{k+1}^{-1}-I)R_k+R_kW_{k+1}.
\end{align}
Then using the decomposition of $T$ (into $T_l$) may yield $\|T_k\|_s=O(\|R_k\|_s).$
As a result, we have
\begin{align}
V_{k+1}^{-1}Q_k^{-1}(T_{k}+D_{k}) Q_kV_{k+1}=Q_k^{-1}(T_{k}+D_{k}) Q_k+R_{(2)},
\end{align}
where
\begin{align}
\nonumber R_{(2)}&=V_{k+1}^{-1}Q_k^{-1}(T_{k}+D_{k}) Q_kV_{k+1}-Q_k^{-1}( T_{k}+D_{k}) Q_k\\
\label{R(2)}&=(V_{k+1}^{-1}-I)Q_k^{-1}(T_{k}+D_{k}) Q_kW_{k+1} \\
\nonumber&\ \ \ +(V_{k+1}^{-1}-I)Q_k^{-1}(T_{k}+D_{k}) Q_k\\
\nonumber&\ \ \ +Q_k^{-1}(T_{k}+D_{k}) Q_kW_{k+1}.
\end{align}
We may assume  prior that $D_{k}=O(\|R_k\|_s)$.
Then it suffices to eliminate terms of order $O(\|R_k\|_s)$. This leads to the following two homological equations
\begin{align}
\label{hme1}&\overline{Q_k^{-1}(D_k+ T_{k})Q_k+R_k}=0,\\
\label{hme2}&[D, W_{k+1}]+S_{\theta_{k+1}}(Q_k^{-1}(T_{k}+D_{k})Q_k+R_k)=0.
\end{align}
In \eqref{hme2} a smoothing operator $S_{\theta_{k+1}}$ is introduced so that one could deal with the \textit{loss of derivatives}.
If we can solve \eqref{hme1} and \eqref{hme2}, then we may obtain
\begin{align*}
Q^{-1}_{k+1}H_{k+1}Q_{k+1}=D+R_{k+1},
\end{align*}
where
\begin{align}\label{Rk1for}
R_{k+1}=(I-S_{\theta_{k+1}})(Q_k^{-1}(T_{k}+D_{k})Q_k+R_k)+R_{(1)}+R_{(2)}
\end{align}
with $\|R_{k+1}\|_s\leq \theta_{k+1}^{s-\alpha}$.

In the following the above arguments  will be made precisely. We first estimate $Q_{l}^{-1}, Q_l$ by using assumptions on $W_l$ ($l\leq k$). Then we determine $D_k$ by solving \eqref{hme1}. Next, we determine $W_{k+1}$ by solving \eqref{hme2}. Finally, we estimate the remainder $R_{k+1}.$

{\bf Estimate of $Q_l^{-1}, Q_l$ ($1\leq l\leq k$)}.

We start with some useful lemmas.
This first one is about the \textit{Tame} estimate on the product of operators.

\begin{lem}\label{tameprod}
Let  \begin{align}\label{C0}
C_0=K_0+K_1(\alpha_0)
\end{align}
with $K_0, K_1(\alpha_0)$ being given by \eqref{k0} and \eqref{k1},  respectively.  {Let $X, X_i\in\mathcal{M}^s$ for  $1\leq i\leq n$.}
Then for  $n\geq 1$ and $s\geq \alpha_0$, we have
\begin{align}
\label{Xin0}\|\prod_{i=1}^n X_i\|_{\alpha_0}&\leq C_0^{n-1}\prod_{i=1}^n \|X_i\|_{\alpha_0},\\
\label{Xins}\|\prod_{i=1}^n X_i\|_s&\leq nC_0^{n}K_1(s)\sum_{i=1}^n\left(\prod_{j\neq i} \|X_j\|_{\alpha_0}\right) \|X_i\|_s.
\end{align}
In particular,
\begin{align}
\label{Xn0}\|X^n\|_{\alpha_0}&\leq C_0^{n-1}\|X\|^n_{\alpha_0},\\
\label{Xns}\|X^n\|_s&\leq n^2C_0^{n}K_1(s)\|X\|^{n-1}_{\alpha_0} \|X\|_s.
\end{align}

\end{lem}

\begin{proof}
We first show that for all $n\geq 1$, \eqref{Xin0} holds true.
We prove it by induction. Using Lemma \ref{tame} with $s=\alpha_0$ yields for $n=2$,
\begin{align*}
\|X_1X_2\|_{\alpha_0}&\leq (K_0+K_1{(\alpha_0)})\|X_1\|_{\alpha_0}\|X_2\|_{\alpha_0}\\
&=C_0\|X_1\|_{\alpha_0}\|X_2\|_{\alpha_0}.
\end{align*}
Assume that
\begin{align}\label{xine1}
&\|\prod_{i=1}^n X_i\|_{\alpha_0}\leq C_0^{n-1}\prod_{i=1}^n \|X_i\|_{\alpha_0}.
\end{align}
Then by using Lemma \ref{tame} again, we obtain since \eqref{xine1} that
\begin{align*}
\|\prod_{i=1}^{n+1} X_i\|_{\alpha_0}&\leq K_0\|X_{n+1}\|_{\alpha_0}\|\prod_{i=1}^n X_i\|_{\alpha_0}+K_1(\alpha_0)\|\prod_{i=1}^n X_i\|_{\alpha_0}\|X_{n+1}\|_{\alpha_0}\\
&\leq C_0\|\prod_{i=1}^n X_i\|_{\alpha_0}\|X_{n+1}\|_{\alpha_0}\\
&\leq C_0^{n}\prod_{i=1}^{n+1}\|X_i\|_{\alpha_0}.
\end{align*}
This proves \eqref{Xin0} and then \eqref{Xn0}.

Next,  we try to prove \eqref{Xins}. We also prove it by induction. For $n=2$, we have by Lemma \ref{tame}, $K_0>1, K_1>1$ that
\begin{align*}
\|X_1X_2\|_{s}\leq C_0K_{1}(s)(\|X_1\|_{\alpha_0}\|X_2\|_s+\|X_1\|_{s}\|X_2\|_{\alpha_0}).
\end{align*}
Assume that
\begin{align}\label{xine2}
\|\prod_{i=1}^n X_i\|_s\leq nC_0^{n}K_1(s)\sum_{i=1}^n\left(\prod_{j\neq i} \|X_j\|_{\alpha_0}\right) \|X_i\|_s.
\end{align}
Then by Lemma \ref{tame}, \eqref{Xin0} and \eqref{xine2}, we get
\begin{align*}
\|\prod_{i=1}^{n+1} X_i\|_s&\leq K_0\|X_{1}\|_{\alpha_0}\|\prod_{i=2}^{n+1} X_i\|_s+K_1(s)\|X_{1}\|_s\|\prod_{i=2}^{n+1} X_i\|_{\alpha_0}\\
&\leq nK_0C_0^{n}K_1(s)\sum_{i=2}^{n+1}\left(\prod_{2\leq j\neq i\leq n+1} \|X_j\|_{\alpha_0}\right) \|X_i\|_s\|X_{1}\|_{\alpha_0}\\
&\ \ \ +C_0^{n-1}K_1(s)\prod_{i=2}^{n+1} \|X_i\|_{\alpha_0}\|X_{1}\|_s\\
&\leq \left(nC_0^{n+1}K_1(s)+C_0^{n-1}K_1(s)\right)\sum_{i=1}^{n+1}\left(\prod_{1\leq j\neq i\leq n+1} \|X_j\|_{\alpha_0}\right) \|X_i\|_s\\
&\leq (n+1)C_0^{n+1}K_1(s)\sum_{i=1}^{n+1}\left(\prod_{1\leq j\neq i\leq n+1} \|X_j\|_{\alpha_0}\right) \|X_i\|_s.
\end{align*}
Obviously, \eqref{Xns} is a direct corollary of \eqref{Xins}.

The proof is finished.
\end{proof}

Now we are ready to estimate $Q_l^{-1}, Q_l$ ($1\leq l\leq k$)
\begin{lem}\label{Qllem}
Assume that
\begin{align}
\label{Theta1}&{\Theta^{\delta/2}\geq 8C_0^2},\\
\label{alpha1}& {\alpha>\alpha_0+\tau+5\delta},
\end{align}
where $C_0$ is given by \eqref{C0}. Then there exists
\begin{align*}
C=C(\delta, \tau, \alpha_0, \alpha_1)>0
\end{align*}
such that the following holds: If $\theta_0\geq C$, then for all $1\leq l\leq k$, we have
\begin{align}
\label{qlinlem}\|Q_l^{-1}\|_s&\leq \theta_{l-1}^{(s-\alpha+\tau+4\delta)_++\delta}\ {\rm for}\  s \in[\alpha_0, \alpha_1], \\
\label{qllem}\|Q_l\|_s&\leq \theta_{l-1}^{(s-\alpha+\tau+4\delta)_++\delta}\ {\rm for}\  s \in[\alpha_0, \alpha_1],\\
\label{qlinqlin1}\|Q_{l}^{-1}-Q_{l-1}^{-1}\|_s&\leq \theta_{l-1}^{s-\alpha+\tau+6\delta}\ {\rm for}\  s \in[\alpha_0, \alpha_1],\\
\label{qlql1}\|Q_l-Q_{l-1}\|_s&\leq \theta_{l-1}^{s-\alpha+\tau+6\delta}\ {\rm for}\  s \in[\alpha_0, \alpha_1],
\end{align}
where $x_+=x$ if $x\geq 0$ and $x_+=0$ if $x<0.$
\end{lem}

\begin{proof}
By induction assumptions, we have for $1\leq l\leq k,$
\begin{align}\label{vlin}
\|V_l^{-1}\|_s\leq 1+2K_1(s)\theta_{l-1}^{s-\alpha+\tau+4\delta},\ \|V_l^{-1}-I\|_s\leq 2K_1(s)\theta_{l-1}^{s-\alpha+\tau+4\delta},
\end{align}
which together with  \eqref{Xins} yields
\begin{align}
\nonumber\|Q_l^{-1}\|_s\leq \|\prod_{i=l}^{1}V_i^{-1}\|_s&\leq lC_0^lK_1(s)\sum_{i=1}^l\left(\prod_{j\neq i}\|V_j^{-1}\|_{\alpha_0}\right)\|V_i^{-1}\|_s\\
\nonumber&\leq lC_0^l2^{l-1}K_1(s)\sum_{i=1}^l\|V_i^{-1}\|_s\ ({\rm since}\ 1+2K_1(\alpha_0)\theta_0^{\alpha_0-\alpha+\tau+4\delta}\leq 2)\\
\label{qlin}&\leq l^2C_0^l2^{l-1}K_1(s)\max_{1\leq i\leq l}\|V_i^{-1}\|_s.
\end{align}
 Let
\begin{align*}
1<C_0\leq \theta_0^{\delta/2}.
\end{align*}
Since \eqref{Theta1}, we have for any $l\geq 2,$
\begin{align}
\nonumber l^2C_0^l2^{l-1}&=(l^{\frac{2}{l-1}}C_0^{\frac{l}{l-1}}2)^{l-1}\leq (8C_0^2)^{l-1}\\
\label{2lthta}&\leq(\theta_0^{\delta/2}\Theta^{\delta/2})^{l-1}=\theta_{l-1}^{\delta/2}.
\end{align}
Since \eqref{qlin} and \eqref{2lthta}, we can obtain for
\begin{align*}
\alpha_0\leq s<\alpha-\tau-4\delta
\end{align*}
that
\begin{align}
\nonumber\|Q_l^{-1}\|_s&\leq l^2C_0^l2^{l-1}K_1(s)(1+2K_1(s)\theta_0^{s-\alpha+\tau+4\delta})\\
\nonumber&\leq l^2C_0^l2^{l-1}K_1(s)(1+2K_1(s))\\
\nonumber&\leq K_1(s)(1+2K_1(s))\theta_{l-1}^{\delta/2} \\
\nonumber&\leq K_1(s)(1+2K_1(s))\theta_0^{-\delta/2}\theta_{l-1}^{\delta}\\
\label{qlinthta1}&\leq \theta_{l-1}^{\delta}\ ({\rm if}\ \theta_0\geq C(\delta, \tau, \alpha_0, \alpha_1)>0).
\end{align}
Similarly,  for
\begin{align*}
\alpha-\tau-4\delta\leq s\leq \alpha_1,
\end{align*}
we obtain
\begin{align}
\nonumber\|Q_l^{-1}\|_s&\leq l^2C_0^l2^{l-1}K_1(s)(1+2K_1(s)\theta_{l-1}^{s-\alpha+\tau+4\delta})\\
\nonumber&\leq 3K_1^2(s)\theta_{l-1}^{\delta/2}\theta_{l-1}^{s-\alpha+\tau+4\delta}\\
\nonumber&\leq 3K_1^2(s)\theta_0^{-\delta/2}\theta_{l-1}^{s-\alpha+\tau+5\delta}\\
\label{qlinthta2}&\leq \theta_{l-1}^{s-\alpha+\tau+5\delta}\ ({\rm if}\ \theta_0\geq C(\delta, \tau, \alpha_0, \alpha_1)>0).
\end{align}
Combining \eqref{qlinthta1} and \eqref{qlinthta2} implies \eqref{qlinlem}. The estimate of $Q_l$ is similar and easier.

We now turn to the differences. Note that
\begin{align*}
Q^{-1}_l-Q_{l-1}^{-1}=(V_l^{-1}-I)Q_{l-1}^{-1},
\end{align*}
which together with  \eqref{tame1}, \eqref{vlin} and \eqref{qlinlem} yields for $l\geq 2$,
\begin{align*}
\|Q^{-1}_l-Q_{l-1}^{-1}\|_s&\leq K_0\|(V_l^{-1}-I)\|_{\alpha_0}\|Q_{l-1}^{-1}\|_s\\
&\ \ \ +K_1(s)\|(V_l^{-1}-I)\|_{s}\|Q_{l-1}^{-1}\|_{\alpha_0}\\
&\leq  2K_0K_1(\alpha_0)\theta_{l-2}^{(s-\alpha+\tau+4\delta)_{+}+\delta}\theta_{l-1}^{\alpha_0-\alpha+\tau+4\delta}\\
&\ \ \ +2K_1(s)\theta_{l-1}^{s-\alpha+\tau+4\delta}\theta_{l-2}^{\delta}\\
&\leq \theta_{l-1}^{s-\alpha+\tau+6\delta},
\end{align*}
where we use the fact $\theta_0\geq C(\delta, \tau, \alpha_0,\alpha_1)>0$ and
\begin{align*}
\theta_{l-2}^{\delta}\theta_{l-1}^{\alpha_0-\alpha+\tau+4\delta}&\leq \theta_{l-1}^{s-\alpha+\tau+6\delta}\theta_0^{\alpha_0-s-\delta}\  {\rm for}\ s-\alpha+\tau+4\delta<0,\\
\theta_{l-1}^{s-\alpha+\tau+4\delta}\theta_{l-2}^{\delta}&\leq \theta_{l-1}^{s-\alpha+\tau+6\delta}\theta_{0}^{-\delta}.
\end{align*}
Finally, if $\theta_0\geq C(\delta, \tau, \alpha_0, \alpha_1)>0$, then
\begin{align*}
\|Q_1^{-1}-Q_0^{-1}\|_s&=\|(I+W_1)^{-1}-I\|_s\\
&\leq 2K_1(s)\theta_0^{s-\alpha+\tau+4\delta}\\
 &\leq \theta_0^{s-\alpha+\tau+6\delta},
\end{align*}
which implies \eqref{qlinqlin1}. The estimate of $Q_l-Q_{l-1}$ (i.e., \eqref{qlql1}) is similar and easier.

This finishes the proof.
\end{proof}

{\bf Estimate of $D_k$}

 Note that if $X\in \mathcal{M}_0^\infty$, then $\|X\|_s=\|X\|_{0}$ for any $s\geq0$. Obviously, $\mathcal{M}^\infty_0$ is a Banach space in the $\|\cdot\|_{\alpha_0}$-norm.
\begin{lem}\label{fixlem}
Let $Q_, \  Q^{-1}\in \mathcal{M}^{\alpha_0}$ satisfy
\begin{align}\label{fixpc}
C_0\|Q-I\|_{\alpha_0}\leq \frac{1}{10},\ C_0\|Q^{-1}-I\|_{\alpha_0}\leq \frac{1}{10},
\end{align}
where $C_0$ is given by \eqref{C0}. Then for any $P,\  P'\in \mathcal{M}^{\alpha_0}$, the equation
\begin{align}\label{hme1lem}
\overline{Q^{-1}(X+P)Q+P'}=0
\end{align}
has a unique solution $X\in \mathcal{M}^\infty_0$ with
\begin{align}\label{fixse}
\|X\|_{\alpha_0}\leq 2(\|Q^{-1}PQ\|_{\alpha_0}+\|P'\|_{\alpha_0}),
\end{align}
where $\overline {X}\in \mathcal{M}_0^\infty$ denotes the diagonal part of some $X\in \mathcal{M}$.
\end{lem}
\begin{proof}
The proof is based on the Banach fixed point theorem. It is easy to see \eqref{hme1lem} is equivalent to
\begin{align}\label{newe}
X-\overline{(Q^{-1}(X+P)Q+P')}=X.
\end{align}
Define a map $f: \mathcal{M}_0^\infty\to \mathcal{M}_0^\infty$  given by
\begin{align*}
f(X)=X-\overline{(Q^{-1}(X+P)Q+P')}.
\end{align*}
To prove the existence and uniqueness of the solution of \eqref{newe}, it suffices to show $f$ is a contractive map. Obviously, we have
\begin{align}\label{yqiny}
Y-Q^{-1}YQ=-(Q^{-1}-I)Y(Q-I)-(Q^{-1}-I)Y-Y(Q-I).
\end{align}
Let $X', X''\in \mathcal{M}_0^\infty$ be arbitrary and $Y=X'-X''$.
This combining \eqref{yqiny} and  Lemma \ref{tame} implies
\begin{align*}
\|f(X')-f(X'')\|_{\alpha_0}&=\|\overline{(Q^{-1}-I)Y(Q-I)+(Q^{-1}-I)Y+Y(Q-I)}\|_{\alpha_0}\\
&\leq \|(Q^{-1}-I)Y(Q-I)+(Q^{-1}-I)Y+Y(Q-I)\|_{\alpha_0}\\
&\leq C_0^2\|Q-I\|_{\alpha_0}\|Q^{-1}-I\|_{\alpha_0}\|Y\|_{\alpha_0}\\
\ \ \ &+C_0\|Q-I\|_{\alpha_0}\|Y\|_{\alpha_0}+C_0\|Q^{-1}-I\|_{\alpha_0}\|Y\|_{\alpha_0}\\
&\leq\frac12\|Y\|_{\alpha_0}= \frac12\|X'-X''\|_{\alpha_0}\ ({\rm since}\ \eqref{fixpc}).
\end{align*}
We have shown that $f$ is a contractive map, and thus the existence and uniqueness of $X$.

Now we estimate the solution $X$. From \eqref{yqiny},  \eqref{fixpc} and Lemma \ref{tame}, we obtain
\begin{align*}
\|X\|_{\alpha_0}&=\|\overline{X-Q^{-1}XQ-Q^{-1}PQ-P'}\|_{\alpha_0}\\
&\leq\|(Q^{-1}-I)X(Q-I)+(Q^{-1}-I)X+X(Q-I)\|_{\alpha_0}\\
\ \ \ &\ \ \ +\|Q^{-1}PQ+P'\|_{\alpha_0}\\
&\leq \frac12\|X\|_{\alpha_0}+\|Q^{-1}PQ\|_{\alpha_0}+\|P'\|_{\alpha_0},
\end{align*}
which implies \eqref{fixse}.

We finish the proof

\end{proof}

Next, we estimate the decomposition $T_k=(S_{\theta_k}-S_{\theta_{k-1}})T.$
\begin{lem}\label{Tklem}
We have that
\begin{align}
\label{Tkls}&\|T_k\|_s\leq \theta_k^{s-s'} \|T\|_{s'}\ {\rm for}\  s\geq s'\geq 0,\\
\label{Tkss}&\| T_k\|_s\leq \theta_{k-1}^{s-s'}\|T\|_{s'}\ {\rm for}\  \alpha_0\leq s\leq s'.
\end{align}
In particular, if
\begin{align}\label{Theta2}
{\|T\|_{s'+\delta}\leq 1,\ \theta_0\geq\Theta^{(s'+\delta-\alpha_0)/\delta}},
\end{align}
then
\begin{align}
\label{Tklsd}\| T_k\|_{s}&\leq \theta_{k}^{s-s'-\delta}\ {\rm for}\  s\geq s'+\delta,\\
\label{Tkssd}\| T_k\|_{s}&\leq \theta_{k}^{s-s'} \ {\rm for}\ s\in[\alpha_0, s'+\delta).
\end{align}
\end{lem}

\begin{proof}
The proof of \eqref{Tkls} and \eqref{Tkss} is trivial and relies on  properties of smoothing operators (i.e.,  \eqref{smooth1} and \eqref{smooth2} of Lemma \ref{smooth}). Now we turn to the proof of \eqref{Tklsd} and \eqref{Tkssd}.
Note that  $\|T\|_{s'+\delta}\leq 1$, which implies \eqref{Tklsd}. While for $\alpha_0\leq s< s'+\delta$, we obatin
\begin{align}
\nonumber \|T_k\|_s&\leq\theta_{k-1}^{s-s'-\delta}=\Theta^{s'+\delta-s}\theta_k^{s-s'-\delta}\\
\nonumber&\leq \theta_0^{-\delta}\Theta^{s'+\delta-\alpha_0}\theta_k^{s-s'} \\
\nonumber&\leq\theta_k^{s-s'},
\end{align}
where in the last inequality we use \eqref{Theta2}.
\end{proof}

Now we estimate  $D_k$.

\begin{lem}\label{Dklem}
Let
\begin{align}
\label{alpha2}&{\kappa=-\alpha+\alpha_0+\tau+6\delta<0},\\
\label{T1}&{\|T\|_{\alpha+4\delta}\leq 1},\\
\label{Theta3}&{\theta_0^\delta\geq \Theta^{\alpha+4\delta-\alpha_0}}.
\end{align}
Then there is a $C=C(\kappa, \delta, \alpha_0)>0$ such that if $\theta_0\geq C$, the homological equation \eqref{hme1} admits a unique solution $D_k\in \mathcal{M}_0^\infty$ satisfying
\begin{align*}
\|D_k\|_{\alpha_0}\leq 3\theta_{k}^{\alpha_0-\alpha}.
\end{align*}

\end{lem}

\begin{proof}
The proof is based on Lemma \ref{Qllem} and  \ref{fixlem}. We set $P= T_k, P'=R_k$ and $Q=Q_k$ in Lemma \ref{fixlem}. It suffices to check the conditions \eqref{fixpc}.
From Lemma \ref{Qllem}, we have since \eqref{qlinqlin1} and \eqref{Theta1} that
\begin{align*}
\|Q_k^{-1}-I\|_{\alpha_0}&\leq \sum_{l=1}^{k}\theta_{l-1}^{\kappa}=\theta_0^{\kappa}\sum_{l=1}^{k} \Theta^{(l-1)\kappa}\\
&\leq \theta_0^{\kappa}\sum_{l=1}^{\infty} (8C_0^2)^{\frac{2\kappa(l-1)}{\delta}}\\
&\leq \theta_0^{\kappa}C(\kappa, \delta, \alpha_0)\leq \frac{1}{10C_0},
\end{align*}
where in the last inequality we use $\theta_0\geq C(\delta, \kappa, \alpha_0)>0.$
Similarly, we obtain $\|Q_k-I\|_{\alpha_0}\leq \frac{1}{10C_0}.$ Thus the conditions of Lemma \ref{fixlem} are verified.

Now let $s'=\alpha+3\delta$ in \eqref{Tkssd} of Lemma \ref{Tklem}. We have since \eqref{T1} and Lemma \ref{Tklem} that
\begin{align*}
\|P\|_s=\|T_k\|_s\leq \theta_{k}^{s-\alpha-3\delta}\ {\rm for}\ s\in [\alpha_0, \alpha+4\delta),
\end{align*}
which together with Lemma \ref{fixlem} implies
\begin{align*}
\|D_k\|_{\alpha_0}&\leq 2(\|Q_k^{-1}PQ_k\|_{\alpha_0}+\|P'\|_{\alpha_0})\\
&\leq 2C_0^2\|Q_k\|_{\alpha_0}\|Q_k^{-1}\|_{\alpha_0}\|P\|_{\alpha_0}+2\theta_k^{\alpha_0-\alpha}\\
&\leq 2C_0^2\theta_{k-1}^{2\delta}\|P\|_{\alpha_0}+2\theta_k^{\alpha_0-\alpha}\\
&\leq 4C_0^2\theta_{k-1}^{2\delta}\theta_k^{\alpha_0-\alpha-3\delta}+2\theta_k^{\alpha_0-\alpha}\\
&= 4C_0^2\theta_0^{-\delta}\theta_k^{\alpha_0-\alpha}+2\theta_k^{\alpha_0-\alpha}\leq 3\theta_k^{\alpha_0-\alpha},
\end{align*}
where in the last inequality we use $\theta_0\geq C(\delta, \alpha_0)>0$.
\end{proof}

{\bf Estimate of $W_{k+1}$}

Now we try to solve the homological equation \eqref{hme2}. The \textit{loss of derivatives} appears in this step. We first establish some useful estimates.
\begin{lem}\label{QDTlem}
Under the assumptions of Lemma \ref{Dklem} and assuming further
\begin{align}\label{alpha3}
{\kappa_1=\alpha_0-\alpha+\tau+7\delta<0},
\end{align}
then we have for $T_k=(S_{\theta_k}-S_{\theta_{k-1}})T,$
\begin{align}
\label{qktkqk}&\|Q_k^{-1}T_kQ_k\|_s\leq \theta_k^{s-\alpha}\ {\rm for}\ s\in [\alpha_0, \alpha_1],\\
\label{qkdkqks}&\|Q_k^{-1}D_kQ_k\|_s\leq \theta_k^{\alpha_0-\alpha+3\delta}\ {\rm for}\ s\in [\alpha_0, \alpha-\tau-4\delta),\\
\label{qkdkqkl}&\|Q_k^{-1}D_kQ_k\|_s\leq \theta_k^{s-\alpha}\ {\rm for}\ s\in [\alpha-\tau-4\delta, \alpha_1].
\end{align}

\end{lem}

\begin{proof}
From  \eqref{Xins}, \eqref{Tklsd} and \eqref{Tkssd}, we obtain for $s\in [\alpha_0, \alpha_1]$,
\begin{align*}
\| Q_k^{-1}T_kQ_k\|_s&\leq 3C_0^3\|Q_k^{-1}\|_{\alpha_0}\|Q_k^{-1}\|_{\alpha_0}\|T_k\|_s\\
&\ \ \ +3C_0^3\|Q_k\|_{\alpha_0}\|T_k\|_{\alpha_0}\|Q_k^{-1}\|_s\\
&\ \ \ +3C_0^3\|Q_k^{-1}\|_{\alpha_0}\|T_k\|_{\alpha_0}\|Q_k\|_s\\
&\leq 3C_0^3\theta_k^{s-\alpha-3\delta}\theta_{k-1}^{2\delta}+6C_0^3\theta_k^{\alpha_0-\alpha-3\delta}\theta_{k-1}^{(s-\alpha+\tau+4\delta)_++2\delta}\\
&\leq 3C_0^3\theta_k^{s-\alpha-\delta}+6C_0^3\theta_{k}^{(s-\alpha+\tau+4\delta)_++\alpha_0-\alpha-\delta}\\
&\leq 3C_0^3\theta_k^{s-\alpha-\delta}+6C_0^3\theta_{k}^{s-\alpha-\delta}\   ({\rm since} \ \kappa_1<0)\\
&\leq 9C_0^3\theta_0^{-\delta}\theta_k^{s-\alpha}\leq\theta_k^{s-\alpha} \ ({\rm if}\ \theta_0\geq C(\delta, \alpha_0)>0),
\end{align*}
which implies \eqref{qktkqk}. Similarly, we obtain
\begin{align}
\nonumber\|Q_k^{-1}D_kQ_k\|_s&\leq 3C_0^3\|Q_k^{-1}\|_{\alpha_0}\|Q_k\|_{\alpha_0}\|D_k\|_s\\
\nonumber&\ \ \ +3C_0^3\|Q_k\|_{\alpha_0}\|D_k\|_{\alpha_0}\|Q_k^{-1}\|_s\\
\nonumber&\ \ \ +3C_0^3\|Q_k^{-1}\|_{\alpha_0}\|D_k\|_{\alpha_0}\|Q_k\|_s\\
\nonumber&\leq 9C_0^3\theta_k^{\alpha_0-\alpha}\theta_{k-1}^{2\delta}+18C_0^3\theta_k^{\alpha_0-\alpha}\theta_{k-1}^{(s-\alpha+\tau+4\delta)_++2\delta}\\
\label{qkindkqk}&\leq 9C_0^3\theta_k^{\alpha_0-\alpha+2\delta}+18C_0^3\theta_{k}^{(s-\alpha+\tau+4\delta)_++\alpha_0-\alpha+2\delta}.
\end{align}
Then if $\alpha_0\leq s<\alpha-\tau-4\delta$, we obtain since \eqref{qkindkqk} that
\begin{align*}
\|Q_k^{-1}D_kQ_k\|_s\leq 27C_0^3\theta_0^{-\delta}\theta_k^{\alpha_0-\alpha+3\delta}\leq \theta_k^{\alpha_0-\alpha+3\delta}\ {\rm if}\ \theta_0\geq C(\delta, \alpha_0)>0.
\end{align*}
If $\alpha-\tau-4\delta\leq s\leq \alpha_1$, then
\begin{align}
\nonumber\alpha_0-\alpha+2\delta&=s-\alpha+(-s+\alpha_0+2\delta)\\
\nonumber&\leq s-\alpha+(\alpha_0-\alpha+\tau+6\delta)\\
\label{appk1}&\leq s-\alpha-\delta,\\
 \nonumber(s-\alpha+\tau+4\delta)_++\alpha_0-\alpha+2\delta&=s-\alpha+(\alpha_0-\alpha+\tau+6\delta)\\
 \label{appk2}&\leq s-\alpha-\delta,
\end{align}
where in \eqref{appk1} and \eqref{appk2} we use the fact $\kappa_1<0.$ Hence if $\alpha-\tau-4\delta\leq s\leq \alpha_1$, we have 
\begin{align*}
\|Q_k^{-1}D_kQ_k\|_s\leq 27C_0^3\theta_0^{-\delta}\theta_k^{s-\alpha}\leq \theta_k^{s-\alpha}\ {\rm if}\ \theta_0\geq C(\delta, \alpha_0)>0.
\end{align*}

This proves \eqref{qkdkqks} and \eqref{qkdkqkl}.

\end{proof}

We are ready to solve the homological equation \eqref{hme2}. Write
\begin{align}
G=Q_k^{-1}(T_k+D_k)Q_k+R_k.
\end{align}
Since Lemma \ref{Dklem}, we have $\overline{G}=0$. By Lemma \ref{QDTlem}, we obtain
\begin{align}\label{Gs}
\|G\|_s\leq 3\theta_k^{s-\alpha+3\delta}\ {\rm for}\ s\in[\alpha_0, \alpha_1].
\end{align}

\begin{lem}\label{Wk1lem}
The homological equation \eqref{hme2} admits a unique solution $W=W_{k+1}$ which is given by
\begin{align}\label{Wk1for}
W_{\mathbf{i},\mathbf{j}}=\frac{(S_{\theta_{k+1}}G)_{\mathbf{i},\mathbf{j}}}{d_{\mathbf{j}}-d_\mathbf{i}}\ {\rm for}\ \mathbf{i}\neq \mathbf{j},\ \overline{W}=0.
\end{align}
Under the assumptions of Lemma \ref{QDTlem} and assuming further
\begin{align}\label{Theta4}
{\theta_0^{\delta}\geq3\gamma^{-1} \Theta^\tau},
\end{align}
then
\begin{align}\label{Wk1}
\|W_{k+1}\|_s\leq \theta_k^{s-\alpha+\tau+4\delta}\ {\rm for}\ s\in [\alpha_0, \alpha_1].
\end{align}
\end{lem}

\begin{proof}
Note that $[D, W]_{\mathbf{i},\mathbf{j}}=(d_\mathbf{i}-d_\mathbf{j})W_{\mathbf{i},\mathbf{j}}$,  $d_\mathbf{i}-d_\mathbf{j}\neq 0$ for $\mathbf{i}\neq \mathbf{j}$,  and $\overline{G}=0$. Then \eqref{Wk1for} holds true.

We then turn to the estimate. Recall that $h=(d_\mathbf{i})_{\mathbf{i}\in\Z^d}$ is a $(\tau, \gamma)$-distal sequence, i.e.,
 \begin{align}\label{distale}
 \|(h-\sigma_\mathbf{p}h)^{-1}\|_{\mathfrak{B}}\leq \gamma^{-1}\langle \mathbf{p}\rangle^{\tau}\ {\rm for}\ \mathbf{p}\neq\mathbf{0}.
 \end{align}
Let $G_\mathbf{p}$ be the $\mathbf{p}$-diagonal of $G$ (i.e., $G_\mathbf{p}(\mathbf{i})=G_{\mathbf{i},\mathbf{i}-\mathbf{p}}$). Then we obtain since \eqref{distale} and  \eqref{sob} that
 \begin{align*}
\|W\|_s^2&\leq\sum_{0<|\mathbf{p}|\leq \theta_{k+1}}\|G_{\mathbf{p}}\|_\mathfrak{B}^2\|(h-\sigma_\mathbf{p}h)^{-1}\|_\mathfrak{B}^{2}\langle \mathbf{p}\rangle^{2s}\\
&\leq \gamma^{-2}\sum_{0<|\mathbf{p}|\leq \theta_{k+1}}\|G_{\mathbf{p}}\|_\mathfrak{B}^2\langle \mathbf{p}\rangle^{2s+2\tau}\\
&=\gamma^{-2}\|S_{\theta_{k+1}}G\|_{s+\tau}^2,
\end{align*}
which implies
\begin{align*}
\|W_{k+1}\|_s\leq \gamma^{-1}\|S_{\theta_{k+1}}G\|_{s+\tau}.
\end{align*}
Thus recalling \eqref{Gs}, we have for $\alpha_0\leq s\leq \alpha_1-\tau$ (i.e., $s+\tau\leq \alpha_1$),
\begin{align*}
\|W_{k+1}\|_s&\leq \|G\|_{s+\tau}\\
 &\leq 3\gamma^{-1}\theta_{k}^{s-\alpha+\tau+3\delta}\\
&\leq3\gamma^{-1}\theta_0^{-\delta} \theta_{k}^{s-\alpha+\tau+4\delta}\\
&\leq \theta_{k}^{s-\alpha+\tau+4\delta}\ ({\rm since}\ \eqref{Theta4}).
\end{align*}
If $\alpha_1-\tau\leq s\leq \alpha_1$ (i.e., $s+\tau\in [\alpha_1, \alpha_1+\tau]$), then we have
\begin{align*}
\|W_{k+1}\|_s&\leq \gamma^{-1}\|S_{\theta_{k+1}}G\|_{s+\tau}\\
 &\leq\gamma^{-1}\theta_{k+1}^{\tau}\|G\|_{s}\\
&\leq3\gamma^{-1}\theta_{k+1}^{\tau} \theta_{k}^{s-\alpha+3\delta}\\
&= 3\gamma^{-1} \Theta^\tau\theta_0^{-\delta} \theta_k^{s-\alpha+\tau+4\delta}\\
&\leq \theta_{k}^{s-\alpha+\tau+4\delta}\ ({\rm since}\ \eqref{Theta4}).
\end{align*}

This proves \eqref{Wk1}.
\end{proof}

{\bf Estimate of $V_{k+1}^{-1}$}

We first introduce a perturbation argument.

\begin{lem}\label{neulem}
Let  $C_0$ be given by \eqref{C0} and  $W\in \mathcal{M}^s$  with $s\geq\alpha_0$. If
\begin{align}\label{w0s}
4C_0^2\|W\|_{\alpha_0}\leq 1/2,
\end{align}
then we have that $V=I+W$ is invertible in $\mathcal{M}^s$ and
\begin{align*}
\|V^{-1}\|_s&\leq 1+2K_1(s)\|W\|_s\ (s>\alpha_0),\\
\|V^{-1}\|_{\alpha_0}&\leq 2,\\
\|V^{-1}-I\|_s&\leq 2K_1(s)\|W\|_s.
\end{align*}
where $K_1(s)$ is given by \eqref{k1} and $I$ denotes the identity operator.
\end{lem}

\begin{proof}
The proof is based on the Neumann series argument and Lemma \ref{tameprod}.  From \eqref{Xins}, we have for $n\geq 2,$
\begin{align}
\nonumber\|W^n\|_{s}&\leq n^2C_0^{n}K_1(s)\|W\|^{n-1}_{\alpha_0}\|W\|_s\\
\label{wn0}&\leq (4C_0^{2}\|W\|_{\alpha_0})^{n-1}K_1(s)\|W\|_s,
\end{align}
where we use the fact  $n^{{2}}\leq 4^{n-1}$ and $n\geq 2$. Then by \eqref{w0s} and \eqref{wn0},
\begin{align*}
\|W^n\|_s\leq 2^{-(n-1)}K_1(s)\|W\|_s\ {\rm for}\ n\geq 2,
\end{align*}
which implies
\begin{align*}
\sum_{n=0}^{\infty}\|W^n\|_s&\leq 1+\|W\|_s+\sum_{n\geq 2}2^{-(n-1)}K_1(s)\|W\|_s\\
&\leq 1+2K_1(s)\|W\|_s<\infty,\\
\sum_{n=1}^{\infty}\|W^n\|_s&\leq \|W\|_s+\sum_{n\geq 2}2^{-(n-1)}K_1(s)\|W\|_s\\
&\leq  2K_1(s)\|W\|_s.
\end{align*}
Since $\mathcal{M}^s$ is a Banach space, applying the standard Neumann series argument shows
\begin{align*}
(I+W)^{-1}&=\sum_{n=0}^{\infty}(-W)^n\in\mathcal{M}^s,\ \|(I+W)^{-1}\|_s\leq 1+2K_1(s)\|W\|_s,\\
\|(I+W)^{-1}-I\|_s&\leq 2K_1(s)\|W\|_s,\ \|(I+W)^{-1}\|_{\alpha_0}\leq 2.
\end{align*}

\end{proof}

Then we estimate $V_{k+1}^{-1}$ with $V_{k+1}=I+W_{k+1}$ and $W_{k+1}$ being given by Lemma \ref{Wk1lem}.

\begin{lem}\label{Vk1lem}
Under the assumptions of Lemma \ref{Wk1lem}, we have
\begin{align}
\|V_{k+1}^{-1}\|_s&\leq 1+2K_1(s)\theta_k^{s-\alpha+\tau+4\delta}\ {\rm for} \ s\in [\alpha_0, \alpha_1],\\
\|V_{k+1}^{-1}-I\|_s&\leq 2K_1(s)\theta_k^{s-\alpha+\tau+4\delta}\ {\rm for} \ s\in [\alpha_0, \alpha_1].
\end{align}
\end{lem}
\begin{proof}
This is a direct consequence of Lemma \ref{neulem} and  \ref{Wk1lem}.
\end{proof}

{\bf Estimate of $R_{k+1}$}

Now we estimate the reminders. Recalling \eqref{Rk1for}, we have
\begin{align*}
R_{k+1}&=R'_{k+1}+R_{k+1}'',\\
R_{k+1}'&=(I-S_{\theta_{k+1}})\left( Q_k^{-1}(T_k+D_k)Q_k+R_k\right),\\
R_{k+1}''&=R_{(1)}+R_{(2)},
\end{align*}
where $R_{(1)}, R_{(2)}$ are given by \eqref{R(1)} and \eqref{R(2)}, respectively. Note that $R_{k+1}'$ comes from the smoothing procedure, while $R_{k+1}''$ is the standard Newton error.

We first estimate $R_{k+1}'$.

\begin{lem}\label{Rk1lem1}
Under the assumptions of Lemma \ref{Wk1lem} and assuming further
\begin{align}
\label{alpha11}&\alpha_1\geq 2\alpha+\delta,\\
\label{Theta5}&\max{(\Theta^{-\alpha_0}, \Theta^{-\delta})}\leq1/10,
\end{align}
then for $\theta_0\geq C(\delta, \tau, \alpha_0)>0,$ we have
\begin{align}\label{Rk1e1}
\|R_{k+1}'\|_s\leq \frac{1}{2}\theta_{k+1}^{s-\alpha}\ {\rm for}\ s\in [\alpha_0, \alpha_1].
\end{align}
\end{lem}
\begin{proof}
First, we estimate $\|(I-S_{\theta_{k+1}})R_k\|_s$. We have two cases.
\begin{itemize}
\item[\textbf{Case 1.}] $s\in[\alpha+\delta, \alpha_1].$ In this case we have
\begin{align}
\nonumber\|(I-S_{\theta_{k+1}})R_k\|_s&\leq\|R_k\|_s\leq \theta_k^{s-\alpha}\\
\nonumber&=\Theta^{\alpha-s}\theta_{k+1}^{s-\alpha}\\
\nonumber&\leq \Theta^{-\delta}\theta_{k+1}^{s-\alpha}\\
\label{Rk1ls}&\leq \frac{1}{10}\theta_{k+1}^{s-\alpha}\ ({\rm since }\ \eqref{Theta5}).
\end{align}
\item[\textbf{Case 2.}]$s\in[\alpha_0, \alpha+\delta).$ In this case we have since \eqref{alpha11} that
\begin{align}
\nonumber\|(I-S_{\theta_{k+1}})R_k\|_s&\leq\theta_{k+1}^{-\alpha}\|R_k\|_{s+\alpha}\\
\nonumber&\leq \theta_{k+1}^{-\alpha}\theta_{k}^{s}\\
\nonumber&=\Theta^{-s}\theta_{k+1}^{s-\alpha}\\
\label{Rk1ss}&\leq \Theta^{-\alpha_0}\theta_{k+1}^{s-\alpha}\leq \frac{1}{10}\theta_{k+1}^{s-\alpha}\ ({\rm since }\ \eqref{Theta5}).
\end{align}
\end{itemize}

Next, recalling \eqref{qktkqk}, we can prove similarly
\begin{align*}
\|(I-S_{\theta_{k+1}})Q_k^{-1}T_kQ_k\|_s\leq \frac{1}{10}\theta_{k+1}^{s-\alpha}\ {\rm for}\ s\in [\alpha_0, \alpha_1].
\end{align*}

Finally, for $\|(I-S_{\theta_{k+1}})Q_k^{-1}D_kQ_k\|_s$, we have
\begin{itemize}
\item[\textbf{Case 1.}] $s\in[\alpha+\delta, \alpha_1].$ Similar to the proof of \eqref{Rk1ls},  we have in this case
\begin{align*}
\|(I-S_{\theta_{k+1}})Q_k^{-1}D_kQ_k\|_s\leq\frac{1}{10}\theta_{k+1}^{s-\alpha}.
\end{align*}
\item[\textbf{Case 2.}]$s\in[\alpha-\tau-4\delta, \alpha+\delta).$ Similar to the proof of \eqref{Rk1ss}, we have in this case
\begin{align*}
\|(I-S_{\theta_{k+1}})Q_k^{-1}D_kQ_k\|_s\leq\frac{1}{10}\theta_{k+1}^{s-\alpha}.
\end{align*}
\item[\textbf{Case 3.}]$s\in[\alpha_0, \alpha-\tau-4\delta).$ In this case we have since $s+\alpha\geq \alpha_0+\alpha>\alpha-\tau-4\delta$ that
\begin{align*}
\|(I-S_{\theta_{k+1}})Q_k^{-1}D_kQ_k\|_s&\leq\theta_{k+1}^{-\alpha}\|Q_k^{-1}D_kQ_k\|_{s+\alpha}\\
&\leq \theta_{k+1}^{-\alpha}\theta_{k}^{s}\ ({\rm since }\ \eqref{qkdkqkl})\\
&\leq \Theta^{-\alpha_0}\theta_{k+1}^{s-\alpha}\leq \frac{1}{10}\theta_{k+1}^{s-\alpha}\ ({\rm if }\ \Theta^{-\alpha_0}\leq1/10).
\end{align*}
\end{itemize}

This finishes the proof of \eqref{Rk1e1}.
\end{proof}

Now we estimate $R_{k+1}''$. We have
\begin{lem}\label{Rk1lem2}
Under the assumptions of Lemma \ref{Wk1lem} and assuming further
\begin{align}\label{Theta6}
\theta_0^{-\kappa_1}\geq C(\alpha_0,\alpha_1)\Theta^{\alpha-\alpha_0-\kappa_1},
\end{align}
where $\kappa_1<0$ is given by \eqref{alpha3},
then we have
\begin{align*}
\|R_{k+1}''\|_s\leq \frac12\theta_{k+1}^{s-\alpha}\  {\rm for}\ s\in[\alpha_0, \alpha_1].
\end{align*}
\end{lem}

\begin{proof}
From \eqref{Gs} and Lemma \ref{Wk1lem}, we obtain
\begin{align*}
\|[D, W_{k+1}]\|_s&\leq \|S_{\theta_{k+1} }G\|_s\\
&\leq \|G\|_s\leq 3\theta_k^{s-\alpha+3\delta}\ (s\in[\alpha_0, \alpha_1]),
\end{align*}
which together with Lemma \ref{tame} and Lemma \ref{Vk1lem} implies for $s\in [\alpha_0, \alpha_1]$,
\begin{align*}
\|R_{(1)}\|_s&=\|(V_{k+1}^{-1}-I)[D, W_{k+1}]\|_s\\
&\ \ \ +\|(V_{k+1}^{-1}-I)R_kW_{k+1}\|_s\\
&\ \ \ +\|(V_{k+1}^{-1}-I)R_k\|_s+\|R_kW_{k+1}\|_s\\
&\leq C(s, \alpha_0)\theta_k^{s-2\alpha+\alpha_0+\tau+7\delta}\\
&\ \ \ +C(s, \alpha_0)\theta_k^{s-3\alpha+2\alpha_0+2\tau+8\delta}\\
&\ \ \ +C(s, \alpha_0)\theta_k^{s-2\alpha+\alpha_0+\tau+4\delta}\\
&\leq C(s, \alpha_0)\theta_k^{s-\alpha+\kappa_1} \ ({\rm since}\ \kappa_1<0)\\
&\leq C(s, \alpha_0)\Theta^{-(s-\alpha+\kappa_1)}\theta_{k+1}^{s-\alpha+\kappa_1}\\
&\leq C(\alpha_1, \alpha_0) \Theta^{\alpha-\kappa_1-\alpha_0}\theta_{0}^{\kappa_1}\theta_{k+1}^{s-\alpha}\\
&\leq \frac14\theta_{k+1}^{s-\alpha}\ ({\rm since}\ \eqref{Theta6}).
\end{align*}
Similarly, we have by recalling
\begin{align*}
\|Q_k^{-1}(T_{k}+D_{k})Q_k\|_s\leq 2\theta_k^{s-\alpha+3\delta}\  {\rm for}\ s\in[\alpha_0, \alpha_1]
\end{align*}
that
\begin{align*}
\|R_{(2)}\|_s&=\|(V_{k+1}^{-1}-I)Q_k^{-1}(T_{k}+D_{k})Q_kW_{k+1}\|_s\\
&\ \ \ +\|(V_{k+1}^{-1}-I)Q_k^{-1}(T_{k}+D_{k})Q_k\|_s\\
&\ \ \ +\|Q_k^{-1}(T_{k}+D_{k})Q_kW_{k+1}\|_s\\
&\leq C(s, \alpha_0)\theta_k^{s-3\alpha+2\alpha_0+2\tau+11\delta}\\
&\ \ \ +C(s, \alpha_0)\theta_k^{s-2\alpha+\alpha_0+\tau+7\delta}\\
&\leq C(s, \alpha_0)\theta_k^{s-\alpha+\kappa_1}\  ({\rm since}\ \kappa_1<0)\\
&\leq \frac14\theta_{k+1}^{s-\alpha}.
\end{align*}

This finishes the proof. 
\end{proof}

Combining Lemma \ref{Wk1lem},  \ref{Vk1lem},   \ref{Rk1lem1} and  \ref{Rk1lem2} then completes the proof of Proposition \ref{itprop}.

\end{proof}

\subsection{ The initial step}

Now we turn to the initial step. 
\begin{prop}\label{prop0}
Assume that
\begin{align}
\label{alpha0}&-\alpha+\alpha_0+\tau+3\delta<0,\\
\label{T2}&\|T\|_{\alpha+3\delta}\leq\theta_0^{\alpha_0-\alpha}\leq 1,\\
\label{Theta0}&\theta_0^{\delta}\geq C(\alpha_0, \alpha_1)\Theta^{\alpha-\alpha_0+\delta}.
\end{align}
Then there is some $C=C(\delta, \tau, \gamma, \alpha_0, \alpha_1)>0$ such that for  $\theta_0\geq C$,
there exist $W_1,  R_1\in \mathcal{M}$ with $V_1=I+W_1$  such that
\begin{align}\label{H1}
V_1^{-1}(T_0+D)V_1=D+R_1,
\end{align}
where
\begin{align}
\label{W1s}\|W_1\|_s &\leq \theta_0^{s-\alpha+\tau+\delta}\ {\rm for}\ s\in [\alpha_0, \alpha_1],\\
\label{V1ins}\|V_1^{-1}\|_s&\leq  \theta_0^{(s-\alpha+\tau+\delta)_++\delta}\ {\rm for}\ s\in [\alpha_0, \alpha_1],\\
\label{V1ins1}\|V_1^{-1}-I\|_s&\leq \theta_0^{s-\alpha+\tau+2\delta}\ {\rm for}\ s\in [\alpha_0, \alpha_1],\\
\label{R1s}\|R_1\|_s&\leq \theta_1^{s-\alpha}\ {\rm for}\ s\in [\alpha_0, \alpha_1].
\end{align}

\end{prop}
\begin{proof}
 Note that $\overline{T}=0$. Then
\begin{align*}
[D, W_1]+ T_0=0
\end{align*}
has a unique solution $W_1.$  We also have
\begin{align}
\nonumber V_1^{-1} (T_0+D)V_1-D&=V_1^{-1}([D, W_1]+T_0V_1)\\
\label{R1for}&=V_1^{-1}T_0 W_1=R_1,
\end{align}
which yields \eqref{H1}.

Next, we estimate $W_1, V_1, R_1$. Because of $T_0=S_{\theta_0}T$,  we obtain for $s\geq \alpha+3\delta,$
\begin{align}\label{T0l}
\| T_0\|_s\leq \theta_0^{s-\alpha-3\delta}\|T\|_{\alpha+3\delta}\leq \theta_0^{s-\alpha-3\delta}\ ({\rm since}\ \eqref{T2}).
\end{align}
If $s<\alpha+3\delta$, we also have by \eqref{T2} that
\begin{align*}
\|T_0\|_s&\leq \|T\|_{\alpha+3\delta}\leq \theta_0^{\alpha_0-\alpha}\leq \theta_0^{s-\alpha}\ ( {\rm since}\ s\geq \alpha_0),
\end{align*}
which together with \eqref{T0l} implies
\begin{align}\label{T0s}
\|T_0\|_s\leq \theta_0^{s-\alpha}\ {\rm for}\ s\in [\alpha_0, \alpha_1].
\end{align}
Similar to the proof of Lemma \ref{Wk1lem}, we obtain
\begin{align*}
\|W_1\|_s\leq \gamma^{-1}\varepsilon \|T_0\|_{s+\tau},
\end{align*}
which implies for $s+\tau\in[\alpha_0+\tau, \alpha_1]$ and $\theta_0\geq C(\gamma, \delta)>0$,
\begin{align*}
\|W_1\|_s&
\leq \gamma^{-1}\theta_0^{-\delta}\theta_0^{s-\alpha+\tau+\delta}\\
&\leq \theta_0^{s-\alpha+\tau+\delta}.
\end{align*}
If $s+\tau\in[\alpha_1, \alpha_1+\tau]$,
we have since \eqref{T2} that
\begin{align*}
\|W_1\|_s
&\leq \gamma^{-1}\theta_0^{\tau}\|T_0\|_{s}\\
&\leq  \gamma^{-1}\theta_0^{\tau+s-\alpha}\\
&\leq \theta_0^{s-\alpha+\tau+\delta}.
\end{align*}
This proves \eqref{W1s}. Now, since Lemma \ref{neulem}, we have for $\theta_0\geq C(\delta, \tau, \alpha_0, \alpha_1)>0,$
\begin{align*}
\|V_1^{-1}\|_s&\leq 1+2K_1(s)\theta_0^{s-\alpha+\tau+\delta}\leq \theta_0^{(s-\alpha+\tau+\delta)_++\delta},\\
\|V_1^{-1}-I\|_s&\leq 2K_1(s)\theta_0^{s-\alpha+\tau+\delta}\leq \theta_0^{s-\alpha+\tau+2\delta},
\end{align*}
which yields \eqref{V1ins} and \eqref{V1ins1}. Finally, we estimate $R_1$. From \eqref{R1for}, \eqref{T0s}, \eqref{V1ins} and Lemma \ref{tame}, we have
\begin{align}
\nonumber\|R_1\|_s&=\|V_1^{-1}T_0W_1\|_s\\
\nonumber&\leq C(\alpha_0, s) \theta_0^{(s-\alpha+\tau+\delta)_+-2\alpha+2\alpha_0+\tau+2\delta} \\
\label{R1se}&\ \ \ +C(\alpha_0, s)\theta_0^{s-2\alpha+\alpha_0+\tau+2\delta}.
\end{align}
If $\alpha_0\leq s<\alpha-\tau-\delta$, then we obtain since \eqref{alpha0} that
\begin{align*}
(s-\alpha+\tau+\delta)_+-2\alpha+2\alpha_0+\tau+2\delta&=-2\alpha+2\alpha_0+\tau+2\delta\\
&=s-\alpha+(-\alpha+2\alpha_0+\tau+2\delta-s)\\
&\leq s-\alpha+(-\alpha+\alpha_0+\tau+2\delta) \ ({\rm since}\ -s\leq -\alpha_0)\\
&\leq s-\alpha-\delta.
\end{align*}
If $\alpha-\tau-\delta\leq s\leq \alpha_1$, we also have by \eqref{alpha0} that
\begin{align*}
(s-\alpha+\tau+\delta)_+-2\alpha+2\alpha_0+\tau+2\delta&=s-3\alpha+2\alpha_0+2\tau+3\delta\\
&=s-\alpha+(-2\alpha+2\alpha_0+2\tau+3\delta)\\
&\leq s-\alpha-3\delta.
\end{align*}
Thus recalling \eqref{R1se}, we have
\begin{align*}
\|R_1\|_s&\leq C(\alpha_0, s) \theta_0^{s-\alpha-\delta} \\
&=C(\alpha_0, s)\Theta^{-(s-\alpha-\delta)}\theta_1^{s-\alpha-\delta}\\
&\leq C(\alpha_0, \alpha_1)\Theta^{\alpha+\delta-\alpha_0}\theta_0^{-\delta}\theta_1^{s-\alpha}\\
&\leq \theta_1^{s-\alpha}\ ({\rm since}\ \eqref{Theta0}).
\end{align*}
This proves \eqref{R1s}.

\end{proof}

\section{The iteration theorem}
In this section we combine Proposition \ref{itprop} and \ref{prop0} to  establish the iteration theorem.

Fix any $\delta>0, \alpha_0>d/2, \tau>0, \gamma>0$. We  collect the conditions imposed on parameters $\Theta, \alpha, \alpha_1, \theta_0$.
Then
\begin{itemize}
\item We let (recalling \eqref{Theta1} and \eqref{Theta5})
\begin{align*}
\Theta=\Theta(\delta, \alpha_0)=\max\{8^{2/\delta}C_0^{4/\delta}, 10^{1/\delta}, 10^{1/\alpha_0} \},
\end{align*}
where $C_0=C_0(\alpha_0)$ is given by \eqref{C0}.
\item We let (recalling \eqref{alpha3})
\begin{align*}
\alpha>\alpha_0+\tau+7\delta.
\end{align*}
\item We let (recalling \eqref{alpha11})
\begin{align*}
\alpha_1\geq 2\alpha+\delta.
\end{align*}
\item We let (recalling \eqref{T1} and  \eqref{T2})
\begin{align*}
\|T\|_{\alpha+4\delta}\leq\theta_0^{\alpha_0-\alpha}\leq 1.
\end{align*}
\item We also assume $\theta_0$ is large enough, i.e.,
\begin{align*}
\theta_0\geq C(\delta, \tau, \gamma, \alpha, \alpha_0, \alpha_1)>0.
\end{align*}
\end{itemize}

\begin{thm}\label{itthm}
Fix $\delta>0, \alpha_0>d/2, \tau>0, \gamma>0$. Assume that
\begin{align*}
&\alpha>\alpha_0+\tau+7\delta,\\
&\alpha_1\geq 2\alpha+\delta,\\
&\Theta=\Theta(\delta, \alpha_0)=\max\{8^{2/\delta}C_0^{4/\delta}, 10^{1/\delta}, 10^{1/\alpha_0} \},\\
&\|T\|_{\alpha+4\delta}\leq \theta_0^{\alpha_0-\alpha}\leq 1.
\end{align*}
Then there is some $C=C(\delta, \tau, \gamma, \alpha, \alpha_0, \alpha_1)>0$ such that for $\theta_0\geq C$,  the following holds true:  For any $k\geq1$ and $1\leq l\leq k$, there exist  $W_l, R_l, Q_l\in \mathcal{M}$ and $D_{l-1}\in \mathcal{M}_0^\infty$  so that
\begin{align}
\nonumber Q_l^{-1}H_lQ_l=D+R_l,\ H_l= \sum_{i=1}^l(T_{i-1}+D_{i-1})+D
\end{align}
with
\begin{align}
\nonumber Q_0&=I, D_0=0, T_0=S_{\theta_0}T,\\
 \nonumber Q_l&=Q_{l-1} V_l, V_l=I+W_l,\\
\nonumber\|W_l\|_{s}&\leq \theta_{l-1}^{s-\alpha+\tau+4\delta}\ {\rm for}\ s\in[\alpha_0, \alpha_1],\\
\nonumber\|V_l^{-1}-I\|_{s}&\leq 2K_1(s)\theta_{l-1}^{s-\alpha+\tau+4\delta}\ {\rm for}\ s\in[\alpha_0, \alpha_1],\\
\nonumber \|R_{l}\|_s&\leq \theta_{l}^{s-\alpha}\ {\rm for}\ s\in[\alpha_0, \alpha_1],\\
\label{Dl0}\|D_{l-1}\|_0&\leq 3\theta_{l-1}^{\alpha_0-\alpha}, 
\end{align}
 where $K_1(s), T_l$ are defined by \eqref{k1} and \eqref{Tl} respectively.

\end{thm}
\begin{proof}
The proof is  based on a combination of Proposition \ref{itprop} and  \ref{prop0}.
\end{proof}

\section{ Proofs of Theorem \ref{mthm1} and  \ref{mthm2}}

\begin{proof}[ Proof of Theorem \ref{mthm1}]

It suffices to prove the convergence of
\begin{align*}
Q_k,\  Q_k^{-1},\  \sum_{i=1}^kD_{l-1},\ \sum_{i=1}^kT_{l-1}\ ({\rm as}\  k\to\infty)
\end{align*}
appeared in Theorem \ref{itthm}.
We set
\begin{align*}
\|T\|_{\alpha+4\delta}=\theta_0^{\alpha_0-\alpha},
\end{align*}
i.e.,
\begin{align*}
\theta_0^{-\delta}=\|T\|_{\alpha+4\delta}^{\frac{\delta}{\alpha-\alpha_0}}.
\end{align*}

We first prove the convergence and estimates of $Q_k$ and $Q_k^{-1}$. From the proof of Lemma \ref{Qllem}, we have
\begin{align*}
\|Q_{l}^{-1}-Q_{l-1}^{-1}\|_s\leq \theta_{l-1}^{s-\alpha+\tau+6\delta}\ (l\geq 1).
\end{align*}
Then
\begin{align*}
\sum_{l\geq 1}\|Q_{l}^{-1}-Q_{l-1}^{-1}\|_{\alpha-\tau-7\delta}&\leq \sum_{l\geq 1}\theta_{l-1}^{-\delta}\leq \theta_0^{-\delta}\sum_{l\geq 0}\Theta^{-l\delta}<\infty,
\end{align*}
which implies
\begin{align}
\nonumber&Q'=I+\sum_{l\geq1}(Q_l^{-1}-Q_{l-1}^{-1})\in \mathcal{M}^{\alpha-\tau-7\delta},\\
\label{q+in}&\|Q'-I\|_{\alpha-\tau-7\delta}\leq C(\delta, \alpha_0)\theta_0^{-\delta}\leq C\|T\|_{\alpha+4\delta}^{\frac{\delta}{\alpha-\alpha_0}}, \\
\nonumber&\|Q_k^{-1}-Q'\|_{\alpha-\tau-7\delta}\leq \sum_{l\geq k+1}\theta_{l-1}^{-\delta}\to 0\ ({\rm as}\  k\to\infty).
\end{align}
Similarly, we also have there is some $Q_+\in M^{\alpha-\tau-7\delta}$ so that
\begin{align}
\label{q+}&\|Q_+-I\|_{\alpha-\tau-7\delta}\leq C(\delta, \alpha_0)\theta_0^{-\delta}\leq C\|T\|_{\alpha+4\delta}^{\frac{\delta}{\alpha-\alpha_0}}, \\
\nonumber&\|Q_k-Q_+\|_{\alpha-\tau-7\delta}\to 0\ ({\rm as}\  k\to\infty).
\end{align}
It is easy to see $Q'=Q^{-1}_+.$

We then show the convergence of $\sum\limits_{l=1}^{k}D_{l-1}$. Recalling \eqref{Dl0}, we obtain
\begin{align*}
\sum_{l\geq 1}\|D_{l-1}\|_{0}&\leq \sum_{l\geq 1}3\theta_{l-1}^{\alpha_0-\alpha}\leq 3\theta_0^{\alpha_0-\alpha}\sum_{l\geq 0}\Theta^{-(\alpha-\alpha_0)l}\\
&\leq C(\delta, \alpha, \alpha_0)\theta_0^{\alpha_0-\alpha}.
\end{align*}
Then there is some $D_+\in \mathcal{M}_0^{\infty}$ so that
\begin{align*}
&\|D_+\|_{0}\leq C(\delta, \alpha, \alpha_0)\theta_0^{\alpha_0-\alpha}\leq C\|T\|_{\alpha+4\delta}^{\frac{1}{\alpha-\alpha_0}},\\
&\|\sum_{l= 1}^kD_{l-1}-D_+\|_{0}\leq  \sum_{l\geq {k+1}}3\theta_{l-1}^{\alpha_0-\alpha}\to 0\ ({\rm as}\  k\to\infty).
\end{align*}

Considering $\sum\limits_{l= 1}^kT_{l-1}$, we have
\begin{align*}
\sum_{l=1}^kT_{l-1}=S_{\theta_{k-1}}T,
\end{align*}
which implies
\begin{align*}
\|T-\sum_{l=1}^kT_{l-1}\|_{\alpha-\tau-7\delta}&=\|(I-S_{\theta_{k-1}})T\|_{\alpha-\tau-7\delta}\\
&\leq \theta_{k-1}^{-\tau-7\delta}\|T\|_{\alpha}\to 0\ ({\rm as}\ k\to\infty).
\end{align*}
Obviously, we obtain
\begin{align*}
\|R_k\|_{\alpha-7\delta}&\leq\theta_k^{-\tau-7\delta}\to 0\ ({\rm as}\ k\to\infty).
\end{align*}

Next, we will show
\begin{align}\label{qpdp1}
Q_+^{-1}(T+D+D_+)Q_+=D.
\end{align}
As mentioned above, $D$ is not necessary in $\mathcal{M}$. However, \eqref{qpdp1} is equivalent to
\begin{align*}
T+D_+&=Q_+ DQ_+^{-1}-D=-[D, Q_+]Q_+^{-1}.
\end{align*}
We will show  $[D, Q_+]\in\mathcal{M}^{\alpha-\tau-7\delta}.$
Note that
\begin{align*}
 \sum _{l=0}^kT_l+\sum_{l=0}^kD_l=-[D, Q_{k+1}]Q_{k+1}^{-1}+Q_{k+1}R_{k+1}Q_{k+1}^{-1},
\end{align*}
where
\begin{align}\label{bradq}
[D, Q_{l+1}]=[D, Q_l(I+W_{l+1})]=[D, Q_l]V_{l+1}+Q_l[D, W_{l+1}],\ [D, Q_0]=0.
\end{align}
Obviously, $[D, Q_+]\in\mathcal{M}.$  It suffices  to prove
\begin{align*}
\lim_{k\to\infty}\|Q_{k+1}R_{k+1}Q_{k+1}^{-1}\|_{\alpha-\tau-7\delta}=0,\\
\lim_{k\to\infty}\|[D, Q_+-Q_{k+1}]\|_{\alpha-\tau-7\delta}=0.
\end{align*}
It is easy to see
\begin{align*}
\|Q_{k+1}R_{k+1}Q_{k+1}^{-1}\|_{\alpha-\tau-7\delta}\leq C(\delta, \tau, \alpha)\theta_k^{-\tau-5\delta}\to0\ ({\rm as}\ k\to\infty).
\end{align*}
From \eqref{bradq} and Theorem \ref{itthm}, we have
\begin{align*}
\|[D, Q_{l}-Q_{l-1}]\|_{\alpha-\tau-7\delta}&\leq C\|[D, Q_{l-1}]\|_{\alpha-\tau-7\delta}\theta_{l-1}^{(\alpha-\tau-7\delta-\alpha)+\tau+4\delta}\\
&\ \ \ +C\theta_{l-1}^{(\alpha-\tau-7\delta)-\alpha+4\delta}\\
&\leq C\|[D, Q_{l-1}]\|_{\alpha-\tau-7\delta}\theta_{l-1}^{-3\delta}+C\theta_{l-1}^{-3\delta}\\
&\leq\|[D, Q_{l-1}]\|_{\alpha-\tau-7\delta}\theta_{l-1}^{-2\delta}+\theta_{l-1}^{-2\delta}\ ({\rm if}\ \theta_0>C).
\end{align*}
Let $a_l=\|[D, Q_{l}]\|_{\alpha-\tau-7\delta}$. Then $0\leq a_1\leq \theta_0^{-7\delta}$ and
\begin{align*}
a_l&\leq (1+\theta_{l-1}^{-2\delta})a_{l-1}+\theta_{l-1}^{-2\delta}\\
&\leq 2a_{l-1}+\theta_{l-1}^{-2\delta}\ ({\rm since}\ \theta_0>C)\\
&\leq 2^2a_{l-2}+2\theta_{l-1}^{-2\delta}+\theta_{l-2}^{-2\delta}\\
&\leq \cdots\\
&\leq 2^{l-1}a_1+2^{l-2}\theta_{1}^{-2\delta}+2^{l-3}\theta_{2}^{-2\delta}+\cdots+ \theta_{l-1}^{-2\delta}\\
&\leq \theta_{l-1}^{\delta}.
\end{align*}
As a result,
\begin{align*}
\|[D, Q_{l}]-[D, Q_{l-1}]\|_{\alpha-\tau-7\delta}&\leq\theta_{l-1}^{-\delta}+\theta_{l-1}^{-2\delta}\leq 2\theta_{l-1}^{-\delta},
\end{align*}
which implies 
\begin{align*}
\|[D, Q_+-Q_{k+1}]\|_{\alpha-\tau-7\delta}&\leq \sum_{l\geq k+2}\|[D, Q_l]-[D, Q_{l-1}]\|_{\alpha-\tau-7\delta}\\
&\leq 2\sum_{l\geq k+2}\theta_{l-1}^{-\delta}\\
&\leq C\theta_{k+1}^{-\delta}\to 0\ ({\rm as}\ k\to\infty).
\end{align*}
Thus $[D, Q_+]=[D, Q_+]+[D, Q_+-Q_{k+1}]\in \mathcal{M}^{\alpha-\tau-7\delta}$ and \eqref{qpdp1} follows.

It remains  to show if   both $T$ and $D$ are real symmetric, then $Q_+$  can be improved to become a unitary operator. Suppose now that
\begin{align*}
(d_\mathbf{i})_{\mathbf{i}\in\Z^d}\in\R^{\Z^d},\ T^t=T.
\end{align*}
 From \eqref{hme1}, we know that $D_+\in\R^{\Z^d}$.
Thus by taking transpose on both sides of \eqref{qpdp1}, we obtain
\begin{align*}
Q_+ DQ_+^{-1}=D+D_++ T=(Q^t_+)^{-1}DQ^t_+,
\end{align*}
which implies
\begin{align*}
D(Q_+^tQ_+)=(Q_+^tQ_+)D.
\end{align*}
As a result, we have
\begin{align}\label{qiqjdidj}
(Q_+^tQ_+)_{\mathbf{i},\mathbf{j}}(d_\mathbf{i}-d_\mathbf{j})=0.
\end{align}
However, since $(d_\mathbf{i})_{\mathbf{i}\in\Z^d}$ is a $(\tau, \gamma)$-distal sequence, we obtain in particular
\begin{align*}
(d_\mathbf{i}-d_\mathbf{j})\neq0 \ {\rm for}\  \mathbf{i}\neq \mathbf{j},
\end{align*}
which together with \eqref{qiqjdidj} yields
\begin{align*}
(Q_+^tQ_+)_{\mathbf{i},\mathbf{j}}=0 \ {\rm for}\  \mathbf{i}\neq \mathbf{j}.
\end{align*}
Thus we have shown $Q_+^tQ_+\in \mathcal{M}_0^{\infty}$. This means that we can set
\begin{align*}
U=Q_+(Q_+^tQ_+)^{-\frac12}.
\end{align*}
It is easy to check that $U$ is a unitary operator and
\begin{align*}
&\ \ \ U^{-1}(T+D+D_+)U\\
&=(Q_+^tQ_+)^{\frac12}(Q_+^{-1}(T+D+D_+)Q_+)(Q_+^tQ_+)^{-\frac12}\\
&=(Q_+^tQ_+)^{\frac12}D(Q_+^tQ_+)^{-\frac12}\\
&=D,
\end{align*}
where in the last equality we use the fact that $Q_+^tQ_+$ is a diagonal operator.

Finally, we estimate $U^t,  U$.
First, we observe that
\begin{align*}
Q_+^tQ_+&=I+(Q_+^t-I)(Q_+-I)+(Q_+^t-I)+(Q_+-I)\\
&:=I+P.
\end{align*}
From \eqref{q+in} and \eqref{q+}, we have
\begin{align}\label{p015}
\|P\|_{0}\leq C\theta_0^{-\delta}.
\end{align}
Since $Q_+^tQ_+$ is diagonal,  $P$ is diagonal as well. Then we can let $P={\rm diag}_{\mathbf{i}\in\Z^d}(p_\mathbf{i})$. It also follows from
\eqref{p015} that
\begin{align*}
\sup_{\mathbf{i}\in\Z^d}|p_\mathbf{i}|\leq C\theta_0^{-\delta}.
\end{align*}
Consequently, we have for $\theta_0\geq C$,
\begin{align*}
(Q_+^tQ_+)^{-\frac12}&={\rm diag}_{\mathbf{i}\in\Z^d}(\sqrt{1+p_\mathbf{i}})\\
&={\rm diag}_{\mathbf{i}\in\Z^d}\left(1+\frac12 p_\mathbf{i}+O(p_\mathbf{i}^2)\right)\\
&=I+P',
\end{align*}
where $P'\in \mathcal{M}_0^\infty$ satisfies $\|P'\|_{0}\leq C\theta_0^{-\delta} $.
Hence we have since Lemma \ref{tame} that
\begin{align*}
\|U^t-I\|_{\alpha-\tau-7\delta}
&=\|U-I\|_{\alpha-\tau-7\delta}\\
&=\|Q_+(Q_+^tQ_+)^{-\frac12}-I\|_{\alpha-\tau-7\delta}\\
&=\|Q_+-I-Q_+ P'\|_{\alpha-\tau-7\delta}\\
&\leq \|Q_+-I\|_{\alpha-\tau-7\delta}+\|Q_+ P'\|_{\alpha-\tau-7\delta}\\
&\leq C\theta_0^{-\delta}+C(\tau, \alpha, \alpha_0)\|Q_+\|_{\alpha-\tau-7\delta}\|P'\|_{\alpha-\tau-7\delta}\\
&\leq C\theta_0^{-\delta}+C\|Q_+\|_{\alpha-\tau-7\delta}\theta_0^{-\delta}\\
&\leq C\theta_0^{-\delta}\leq C\|T\|_{\alpha+4\delta}^{\frac{\delta}{\alpha-\alpha_0}}.
\end{align*}

This proves Theorem \ref{mthm1}.
\end{proof}

\begin{proof}[ Proof of Theorem \ref{mthm2}]
The proof is similar to that of Theorem \ref{mthm1}. Assume that at the $k$-th iteration step we have
\begin{align*}
Q_k^{-1}\left(\sum_{l=1}^kT_{l-1}+D\right)Q_k=D+ \sum_{l=1}^kD_{l-1}+R_k,
\end{align*}
where $D_{l-1}\in \mathcal{M}_0^{\infty}$ ($1\leq l\leq k$).  We want to find $Q_{k+1}=Q_k(I+W_{k+1})\in \mathcal{M}, \ D_k\in \mathcal{M}_0^\infty$ so that
\begin{align*}
Q_{k+1}^{-1}\left(\sum_{l=1}^{k+1}T_{l-1}+D\right)Q_{k+1}=D+ \sum_{l=1}^{k+1}D_{l-1}+R_{k+1}.
\end{align*}
Note that
\begin{align*}
&\ \ \ Q_{k+1}^{-1}\left(\sum_{l=1}^kT_{l-1}+T_k+D\right)Q_{k+1}\\
&=  (I+W_{k+1})^{-1}\left(D+\sum_{l=1}^kD_{l-1}+R_k\right)(I+W_{k+1})\\
&\ \ \ +Q_{k+1}^{-1}(T_k)Q_{k+1}\\
&=D+ \sum_{l=1}^kD_{l-1}+[D+\sum_{l=1}^kD_{l-1}, W_{k+1}]+R_k'+R_{k+1},
\end{align*}
where $S_{\theta_{k+1}}R_k=R_k'$ and $\|R_k'\|_s=O(\|R_k\|_s)$.
Our aim is to eliminate terms of order $O(\|R_k\|_s)$. Then it needs to solve the new homological equation
\begin{align}\label{hme3}
\left[D+\sum_{l=1}^kD_{l-1}, W_{k+1}\right]+R_k'-\overline{R_k'}=0,
\end{align}
which then implies
\begin{align*}
D_k=\overline{R_k'}.
\end{align*}
As compared with \eqref{hme2}, the main part of \eqref{hme3} becomes  $D+\sum\limits_{l=1}^kD_{l-1}$, rather than $D$!  Fortunately, we have a much stronger assumption, i.e.,
$D+D'\in DC_{\mathfrak{B}}(\tau,\gamma)$ for any $D'\in \mathcal{M}_0^\infty$ with $\|D'\|_0\leq \eta$ ($\eta>0$). It is easy to see
\begin{align*}
D+\sum_{l=1}^kD_{l-1}\in DC_{\mathfrak{B}}(\tau,\gamma)
\end{align*}
if $\theta_0\geq C(\eta)>0.$ As a result, the equation \eqref{hme3} can be solved almost the  same as that of \eqref{hme2}.

Once the equation \eqref{hme3} is solved and estimated, the remaining  issue is just to perform a similar iteration  as that in proving Theorem \ref{mthm1}. Thus we omit the details here.

\end{proof}

\section*{Acknowledgements}

This work was supported by the NSFC  (No. 12271380). The author would like to thank the editor and referees for their helpful suggestions. 

\section*{Data Availability}
The manuscript has no associated data.
\section*{Declarations}
\textbf{Conflict of interest} The  author states
that there is no conflict of interest.

\appendix


\section{}
In this appendix we prove Lemma \ref{tame}.
\begin{proof}[{ Proof of Lemma  \ref{tame}}]
The proof is  standard and is based on the H\"older inequality. For any $X\in \mathcal{M}^s$, recall that
\begin{align*}
X_\mathbf{k}=(X_\mathbf{k}(\mathbf{i}))_{\mathbf{i}\in\Z^d},\  X_\mathbf{k}(\mathbf{i})=X_{\mathbf{i}, \mathbf{i}-\mathbf{k}}.
\end{align*}
We first show if $Z=XY,$ then
\begin{align}\label{Zkfor}
Z_\mathbf{k}=\sum_{\mathbf{j}\in\Z^d}X_\mathbf{j}(\sigma_\mathbf{j}Y_{\mathbf{k}-\mathbf{j}}).
\end{align}
In fact, we have
\begin{align*}
Z_\mathbf{k}(\mathbf{i})=Z_{\mathbf{i},\mathbf{i}-\mathbf{k}}
=\sum_{\mathbf{j}\in\Z^d}X_{\mathbf{i},\mathbf{i}-\mathbf{j}}Y_{\mathbf{i}-\mathbf{j}, \mathbf{i}-\mathbf{k}}=\sum_{\mathbf{j}\in\Z^d}X_\mathbf{j}(\mathbf{i})(\sigma_{\mathbf{j}}Y_{\mathbf{k}-\mathbf{j}})(\mathbf{i}),
\end{align*}
which implies \eqref{Zkfor}.

Next, from \eqref{sob}, we obtain since \eqref{Zkfor}
\begin{align*}
\|Z\|_s^2=\sum_{\mathbf{k}\in\Z^d}\|Z_\mathbf{k}\|_\mathfrak{B}^2\langle \mathbf{k}\rangle^{2s}&=\sum_{\mathbf{k}\in\Z^d}\|\sum_{\mathbf{j}\in\Z^d}X_\mathbf{j}(\sigma_\mathbf{j}Y_{\mathbf{k}-\mathbf{j}})\|_\mathfrak{B}^2\langle \mathbf{k}\rangle^{2s}\\
&\leq \sum_{\mathbf{k}\in\Z^d}\left(\sum_{\mathbf{j}\in\Z^d}\|X_\mathbf{j}\|_\mathfrak{B}\|\sigma_\mathbf{j}Y_{\mathbf{k}-\mathbf{j}}\|_\mathfrak{B}\right)^2\langle \mathbf{k}\rangle^{2s}\\
&= \sum_{\mathbf{k}\in\Z^d}\left(\sum_{\mathbf{j}\in\Z^d}\|X_\mathbf{j}\|_\mathfrak{B}\|Y_{\mathbf{k}-\mathbf{j}}\|_\mathfrak{B}\right)^2\langle \mathbf{k}\rangle^{2s},
\end{align*}
where in the last equality we use the translation invariance of $\mathfrak{B}$. Let $a_\mathbf{k}=\|X_\mathbf{k}\|_\mathfrak{B},\  b_{\mathbf{k}}=\|Y_{\mathbf{k}}\|_\mathfrak{B}$. It suffices to study the sum
$
\left(\sum_{\mathbf{j}\in\Z^d}a_\mathbf{j}b_{\mathbf{k}-\mathbf{j}}\right)^2\langle \mathbf{k}\rangle^{2s}.
$
We have the following two cases.
\begin{itemize}
\item[{\bf Case 1.}] $\mathbf{j}\in \mathcal{I}_\mathbf{k}:=\{\mathbf{j}\in\Z^d:\ {\langle \mathbf{k}\rangle^{2s}}{\langle \mathbf{k}-\mathbf{j}\rangle^{-2s}}\leq 10\}.$ In this case we have
\begin{align}\label{inIk}
{\langle \mathbf{k}\rangle^{2s}}{\langle \mathbf{k}-\mathbf{j}\rangle^{-2s}}\langle \mathbf{j}\rangle^{-2\alpha_0}\leq 10\langle \mathbf{j}\rangle^{-2\alpha_0}.
\end{align}
Hence we have by using the H\"older inequality
\begin{align*}
\sum_{\mathbf{k}\in\Z^d}\left(\sum_{\mathbf{j}\in \mathcal{I}_\mathbf{k}}a_\mathbf{j}b_{\mathbf{k}-\mathbf{j}}\right)^2\langle \mathbf{k}\rangle^{2s}
&\leq\sum_{\mathbf{k}\in\Z^d}\left(\sum_{\mathbf{j}\in \mathcal{I}_\mathbf{k}}a_\mathbf{j}^2\langle \mathbf{j}\rangle^{2\alpha_0}\right)\left(\sum_{\mathbf{j}\in \mathcal{I}_\mathbf{k}}b_{\mathbf{k}-\mathbf{j}}^2\langle \mathbf{j}\rangle^{-2\alpha_0}\right) \langle \mathbf{k}\rangle^{2s}\\
&\leq\|X\|_{\alpha_0}^2\sum_{\mathbf{k}\in\Z^d}\left(\sum_{\mathbf{j}\in \mathcal{I}_\mathbf{k}}b_{\mathbf{k}-\mathbf{j}}^2\langle \mathbf{k}-\mathbf{j}\rangle^{2s}\langle \mathbf{j}\rangle^{-2\alpha_0}\langle \mathbf{k}\rangle^{2s}\langle \mathbf{k}-\mathbf{j}\rangle^{-2s}\right) \\
&\leq10\|X\|_{\alpha_0}^2\sum_{\mathbf{k}\in\Z^d}\left(\sum_{\mathbf{j}\in \mathcal{I}_\mathbf{k}}b_{\mathbf{k}-\mathbf{j}}^2\langle \mathbf{k}-\mathbf{j}\rangle^{2s}\langle \mathbf{j}\rangle^{-2\alpha_0}\right)\ ({\rm since}\ \eqref{inIk})\\
&\leq10\|X\|_{\alpha_0}^2\sum_{\mathbf{j}\in\Z^d}\langle \mathbf{j}\rangle^{-2\alpha_0}\left(\sum_{\mathbf{k}\in \Z^d}b_{\mathbf{k}-\mathbf{j}}^2\langle \mathbf{k}-\mathbf{j}\rangle^{2s}\right)\\
&\leq M_0^2\|X\|_{\alpha_0}^2\|Y\|_{s}^2,
\end{align*}
where $M_0=\sqrt{10\sum\limits_{\mathbf{k}\in\Z^d}\langle \mathbf{k}\rangle^{-2\alpha_0}}<\infty$ since $\alpha_0>d/2.$

\item[{\bf Case 2.}] $\mathbf{j}\notin \mathcal{I}_\mathbf{k}.$ In this case we must have $\mathbf{k}\neq \mathbf{0}.$ Then
\begin{align*}
\langle \mathbf{k}\rangle&>10^{\frac{1}{2s}}\langle \mathbf{k}-\mathbf{j}\rangle>10^{\frac{1}{2s}}|\mathbf{k}-\mathbf{j}|\geq 10^{\frac{1}{2s}}(|\mathbf{k}|-|\mathbf{j}|)\geq 10^{\frac{1}{2s}}(\langle \mathbf{k}\rangle-\langle \mathbf{j}\rangle),
\end{align*}
which yields
\begin{align}\label{ninIk}
\langle \mathbf{j}\rangle^{-2s}\leq  (1-10^{-\frac{1}{2s}})^{-2s}\langle \mathbf{k}\rangle^{-2s}.
\end{align}
Thus using again the H\"older inequality implies
\begin{align*}
\sum_{\mathbf{k}\in\Z^d}\left(\sum_{\mathbf{j}\notin\mathcal{I}_{\mathbf{k}}}a_\mathbf{j}b_{\mathbf{k}-\mathbf{j}}\right)^2\langle \mathbf{k}\rangle^{2s}
&\leq \sum_{\mathbf{k}\in\Z^d}\left(\sum_{\mathbf{j}\in\Z^d}b_{\mathbf{k}-\mathbf{j}}^2\langle \mathbf{k}-\mathbf{j}\rangle^{2\alpha_0}\right)\left(\sum_{\mathbf{j}\notin \mathcal{I}_\mathbf{k}}a_{\mathbf{j}}^2\langle \mathbf{k}-\mathbf{j}\rangle^{-2\alpha_0}\right) \langle \mathbf{k}\rangle^{2s}\\
&\leq\|Y\|_{\alpha_0}^2\sum_{\mathbf{k}\in\Z^d}\left(\sum_{\mathbf{j}\notin \mathcal{I}_\mathbf{k}}a_{\mathbf{j}}^2\langle \mathbf{j}\rangle^{2s}\langle \mathbf{k}-\mathbf{j}\rangle^{-2\alpha_0}\langle \mathbf{j}\rangle^{-2s}\langle \mathbf{k}\rangle^{2s}\right) \\
&\leq(1-10^{-\frac{1}{2s}})^{-2s}\|Y\|_{\alpha_0}^2\sum_{\mathbf{k}\in\Z^d}\left(\sum_{\mathbf{j}\in \Z^d}a_{\mathbf{j}}^2\langle \mathbf{j}\rangle^{2s}\langle \mathbf{k}-\mathbf{j}\rangle^{-2\alpha_0}\right)\ ({\rm since}\ \eqref{ninIk})\\
&\leq(1-10^{-\frac{1}{2s}})^{-2s}\|Y\|_{\alpha_0}^2\sum_{\mathbf{j}\in\Z^d}a_{\mathbf{j}}^2\langle \mathbf{j}\rangle^{2s}\left(\sum_{\mathbf{k}\in \Z^d}\langle \mathbf{k}-\mathbf{j}\rangle^{-2\alpha_0}\right)\\
&\leq M_1^2(s)\|Y\|_{\alpha_0}^2\|X\|_{s}^2,
\end{align*}
where $M_1(s)=(1-10^{-\frac{1}{2s}})^{-s}\sqrt{\sum\limits_{\mathbf{k}\in\Z^d}\langle \mathbf{k}\rangle^{-2\alpha_0}}<\infty$ since $\alpha_0>d/2.$
\end{itemize}
Combining {\bf Case 1} and {\bf Case 2} implies
\begin{align*}
\|XY\|_s&\leq\sqrt{2M^2_0(\alpha_0)\|X\|_{\alpha_0}^2\|Y\|_{s}^2+2M^2_1(s)\|Y\|_{\alpha_0}^2\|X\|_{s}^2}\\
&\leq K_0(\alpha_0)\|X\|_{\alpha_0}\|Y\|_{s}+K_1(s)\|X\|_{s}\|Y\|_{\alpha_0},
\end{align*}
which proves Lemma \ref{tame}, where $K_0=\sqrt{2}M_0,\ K_1(s)=\sqrt{2}M_1(s).$

\end{proof}

\bibliographystyle{alpha} 

 \end{document}